\documentclass[10pt,journal,compsoc]{IEEEtran}

\newtheorem{definition}{Definition}
\newtheorem{strategy}{Strategy}    % gws
\newtheorem{theorem}{Theorem}
\newenvironment{proof}{\begin{IEEEproof}}{\end{IEEEproof}}
\usepackage{multirow}
\usepackage{amsmath}

\ifCLASSOPTIONcompsoc
  \usepackage[nocompress]{cite}
\else
  \usepackage{cite}
\fi

\ifCLASSINFOpdf
   \usepackage[pdftex]{graphicx}
\else

\fi

\usepackage{algorithmic}
\usepackage{algorithm}
\usepackage{array}

  \usepackage[caption=false,font=footnotesize,labelfont=sf,textfont=sf]{subfig}
  \usepackage[caption=false,font=footnotesize]{subfig}

\newcommand\MYhyperrefoptions{bookmarks=true,bookmarksnumbered=true,
	pdfpagemode={UseOutlines},plainpages=false,pdfpagelabels=true,
	colorlinks=true,linkcolor={blue},citecolor={blue},urlcolor={blue},
	pdftitle={Utility-driven Mining of Contiguous Sequences},%<!CHANGE!
	pdfsubject={Typesetting},%<!CHANGE!
	pdfauthor={ChunKai Zhang},%<!CHANGE!
	pdfkeywords={Sequence data, utility mining, sequential pattern, contiguous}}%<^!CHANGE!
\ifCLASSINFOpdf
\usepackage[\MYhyperrefoptions,pdftex]{hyperref}
\else
\usepackage[\MYhyperrefoptions,breaklinks=true,dvips]{hyperref}
\usepackage{breakurl}
\fi

\usepackage{url}

\begin{document}

\title{\huge Utility-driven Mining of Contiguous Sequences}

\author{Chunkai Zhang, Quanjian Dai, Zilin Du, Wensheng Gan, Jian Weng, and Philip S. Yu
	
	\IEEEcompsocitemizethanks{\IEEEcompsocthanksitem Chunkai Zhang, Quanjian Dai, and Zilin Du are with Department of Computer Science and Technology, Harbin Institute of Technology (Shenzhen), Shenzhen 518055, China.
		
	\IEEEcompsocthanksitem Wensheng Gan and Jian Weng are with College of Cyber Security, Jinan University, Guangzhou 510632, China. (E-mail: wsgan001@gmail.com) 

	\IEEEcompsocthanksitem Philip S. Yu is with Department of Computer Science, University of Illinois at Chicago, IL, USA.}

	%\thanks{Manuscript received August 31, 2021; revised August 26, XXXX.}
}

\IEEEtitleabstractindextext{%
\begin{abstract}

Recently, contiguous sequential pattern mining (CSPM) gained interest as a research topic, due to its varied potential real-world applications, such as web log and biological sequence analysis. To date, studies on the CSPM problem remain in preliminary stages. Existing CSPM algorithms lack the efficiency to satisfy users’ needs and can still be improved in terms of runtime and memory consumption. In addition, existing algorithms were developed to deal with simple sequence data, working with only one event at a time. Complex sequence data, which represent multiple events occurring simultaneously, are also commonly observed in real life. In this paper, we propose a novel algorithm, fast utility-driven contiguous sequential pattern mining (FUCPM), to address the CSPM problem. FUCPM adopts a compact sequence information list and instance chain structures to store the necessary information of the database and candidate patterns. For further efficiency, we develop the global unpromising items pruning and local unpromising items pruning strategies, based on sequence-weighted utilization and item-extension utilization, to reduce the search space. Extensive experiments on real-world and synthetic datasets demonstrate that FUCPM outperforms the state-of-the-art algorithms and is scalable enough to handle complex sequence data.

\end{abstract}

\begin{IEEEkeywords}
Sequence data, utility mining, sequential pattern, contiguous.
\end{IEEEkeywords}}

\maketitle

\IEEEdisplaynontitleabstractindextext

\IEEEpeerreviewmaketitle

\IEEEraisesectionheading{
	\section{Introduction}}
\label{sec:introduction}

\IEEEPARstart{I}{n} the era of big data, large volumes of data are produced every day, including a large amount of sequence data. Sequence data consist of a series of elements (also called \textit{itemsets}) that are arranged in a specific order, such as chronological order (e.g., web access logs, vehicle trajectories) and biological order (e.g., DNA sequences). Each element is a set of one or more events (also called \textit{items}). Sequential pattern mining (SPM) technology \cite{agrawal1995mining, srikant1996mining} extracts the useful underlying knowledge contained in sequence data. SPM's objective is to discover all the frequent sequences from a sequence database, as sequential patterns, where the frequency (also known as \textit{support}) of a sequence is measured by its occurrence times in the database, with respect to a user-defined minimum support threshold. It is well-recognized that SPM facilitates a variety of applications, including key-phrase extraction \cite{wang2017keyphrase}, medical early alarm systems \cite{yu2021multi}, outlying behavior detection \cite{wang2020efficient}, and recommendation systems \cite{le2021sequential, bin2019personalized}. Amazon utilizes SPM to recommend products with a slogan "customer  bought something always then bought this."

SPM assumes that high-frequency sequential patterns are meaningful and interesting. However, such an assumption is impractical in several real-world situations. For example, an SPM-based traffic flow analysis system assesses congestion by mining vehicle trajectory databases. Routes containing heavily used trunk roads are discovered as frequent patterns and are regarded as congested. However, the frequent routes are not necessarily congested because trunk roads often contain more lanes to accommodate more vehicles, while side roads are relatively narrow to form traffic jams. In this case, the route frequency is not a suitable measurement for congestion analysis, while a vehicle's average speed while traveling the roads may be better. A retailer developing bundled sales strategies to maximize revenue provides another example of this technology in use. SPM technology can only tell the retail which commodity patterns are hot selling but cannot provide any information about profit. In general, frequency cannot represent importance in some scenarios, and infrequent sequences may contain crucial information. To address this problem, SPM was generalized to obtain a new study field called high-utility sequential pattern mining (HUSPM) \cite{gan2021survey,ahmed2010novel, yin2012uspan}.

HUSPM incorporates the concept of utility that reveals the relative importance of the items. The mining objectives of HUSPM are quantitative sequence databases, wherein each item is associated with an integer called internal utility, which represents the quantities of an item in an itemset. Moreover, each distinct item has an external utility that measures its significance (e.g., unit profit, interest, and satisfaction). HUSPM aims to find all sequences whose utility is no less than the user-defined minimum utility threshold in a sequence database. The discovered sequences are called high-utility sequential patterns (HUSPs) \cite{yin2012uspan}. Unlike SPM, HUSPM can find more valuable patterns because utility is more related to realistic business needs than frequency \cite{zihayat2016distributed}. In recent years, HUSPM has been a popular topic in knowledge discovery in databases (KDD) \cite{chen1996data} due to its practicality. Some studies \cite{wang2016efficiently, gan2020proum, gan2020fast} were conducted to improve the efficiency of mining HUSPs, while others explored the practical application of HUSPM in various domains, including purchase behavior analysis \cite{gan2018extracting}, web access analysis \cite{ahmed2011framework}, and mobile computing \cite{shie2013efficient}.

Although SPM and HUSPM can excavate useful information from databases, they suffer from a large number of discovered sequential patterns, especially when processing a large-scale database with a low threshold. In addition, some patterns may be meaningless in certain scenarios. For example, when mining the vehicle trajectory database, the sequential patterns that maintain the continuity of items with respect to the original trajectories are more meaningful than those that are not contiguous because the former correspond to real-world routes, whereas the latter jump from place to place \cite{bermingham2020mining}. In the task of DNA sequence analysis, ensuring that the nucleotide arrangement in the patterns is consistent with that of DNA sequences is important because skip nucleotides are difficult to interpret \cite{nawaz2021using}. In summary, the inherent continuous order of items in the sequence data is non-negligible under certain circumstances. This problem introduces a new research task, called contiguous sequential pattern mining (CSPM) \cite{chen2007mining}. CSPM aims to discover contiguous sequential patterns (CSPs), in which the items must maintain the adjacent relation that is defined in the sequence data \cite{zhang2015ccspan}. In other words, a CSP must be a continuous sub-sequence of some sequences in the database. For example, in Figure \ref{map}, a person goes to \textit{St. Patrick’s Cathedral} from the \textit{Grand Central Terminal} through the street sequence $<$\textit{Vanderbilt Ave, E 46th St, Madison Ave, E 48th st, 5th Ave, W 51st St}$>$ that is marked in green, and sequential patterns, such as $<$\textit{Vanderbilt Ave, E 46th St, Madison Ave}$>$ and $<$\textit{Madison Ave, E 48th st, 5th Ave}$>$ are CSPs, while $<$\textit{Vanderbilt Ave, E 48th st, W 51st St}$>$ is not a CSP. In addition, the streets in $<$\textit{Vanderbilt Ave, E 48th st, W 51st St}$>$ are not connected directly to each other; such a pattern may not be useful for applications such as path planning. To date, CSPM has attracted much attention and has been successfully applied in many fields \cite{jawahar2018efficient, yang2018mining}. 

\begin{figure}[htbp]
	\centering
	\includegraphics[width=0.8\linewidth]{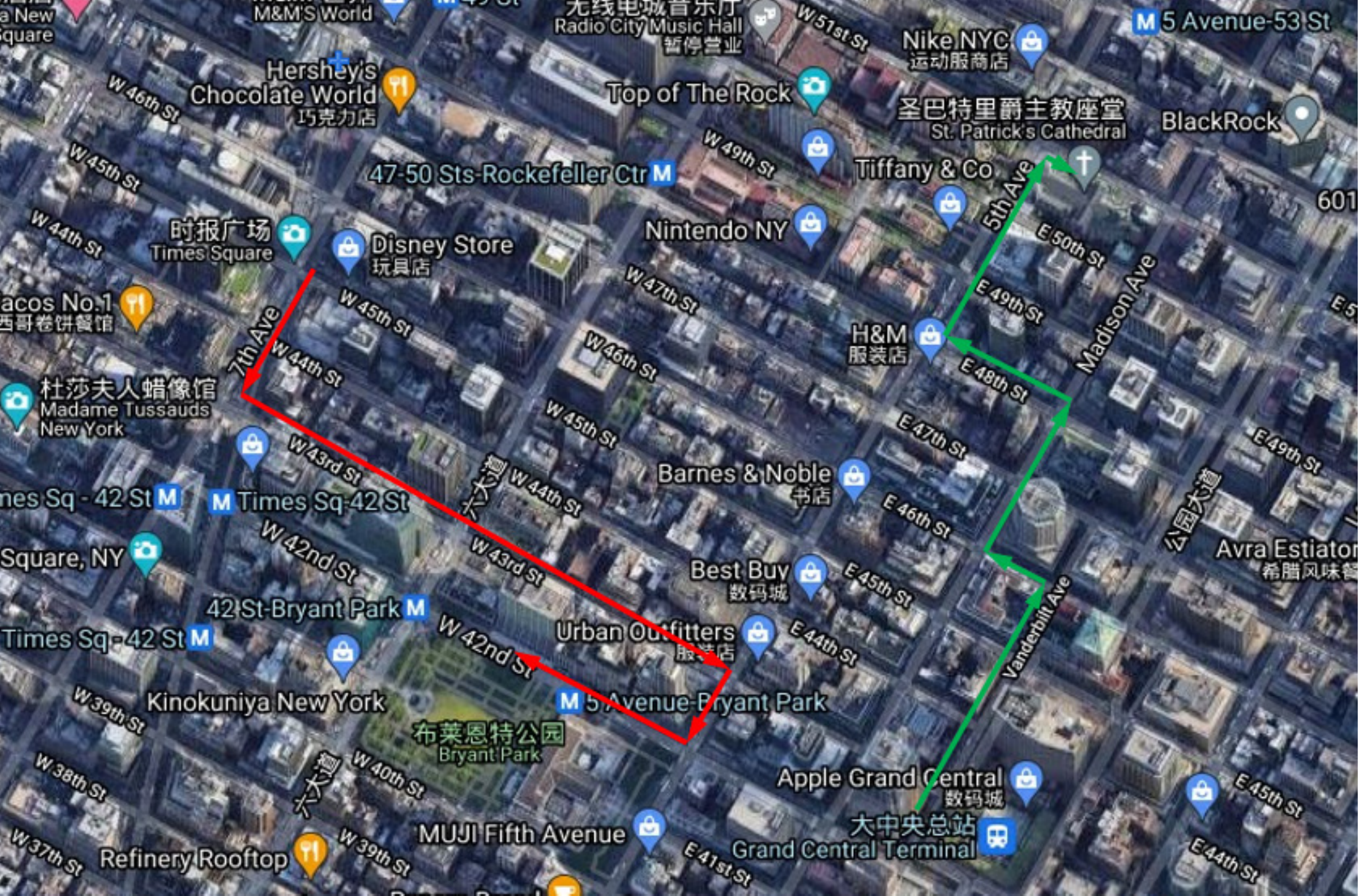}
	\caption{Example of vehicle trajectories}
	\label{map}
\end{figure}

Previous studies on CSPM are mostly based on frequency, while only a few studies \cite{zhou2007utility, ahmed2009efficient, huang2016efficient} have been conducted to address the problem of utility-driven contiguous sequential pattern mining (UCSPM). UCSPM aims to discover high-utility contiguous sequential patterns (HUCSPs) from sequence databases. In fact, UCSPM can be used in a wide range of applications. For example, in \cite{zhou2007utility} and \cite{ahmed2009efficient}, HUCSPs were discovered from users' web browsing sequences, where each web page was regarded as an item and the browsing time was regarded as utility. The discovered patterns can reveal users' browsing preferences and help provide proper navigation suggestions for users. Because HUCSPs maintain the adjacent relationship between the web pages browsed by users, they can also help optimize the website architecture, for example, providing efficient access between highly correlated web pages \cite{ahmed2009efficient}. Other potential applications of UCSPM may include vehicle trajectory analysis \cite{bermingham2020mining} and next-items recommendation \cite{yap2012effective}.

To date, research on UCSPM remains in its early stages. It is noteworthy that almost all existing CSPM and UCSPM methods can only handle simple sequences whose itemsets contain one single item \cite{huang2016efficient}. This type of sequence is called single-item-based sequences. In the age of big data, the volume and complexity of data grow rapidly. Complex sequence data, whose itemsets may contain multiple items (also called multi-items-based sequences), are also commonly observed in real-life scenarios where multiple events occur simultaneously. Considering web browsing sequences, a user commonly browses multiple e-commerce websites to compare price; thus, each itemset in the browsing sequence can have several items. Another noticeable problem is that existing algorithms are not efficient enough to satisfy the requirements of users to quickly discover useful information in large-scale databases. 

In general, several challenges remain in this field, which can be described as follows. First, existing UCSPM approaches \cite{zhou2007utility , ahmed2009efficient, huang2016efficient} only consider single-item-based sequence data but none of them can handle multi-items-based sequence data. Therefore, developing an efficient algorithm that is scalable to large-scale and multi-items-based sequences is an urgent need. Second, UCSPM can easily encounter the combinatorial explosion of the search space. This is because the inherent ordering of itemsets in sequence data generates various compositions of candidate patterns. In particular, when handling complex multi-items-based sequences, the number of candidates can be extremely large. The integration of the contiguous constrain to design powerful pruning strategies is crucial.  Third, the downward closure property \cite{agrawal1994fast} of frequency does not hold for utility; that is, no exact relationship exists between the utility of a sequence and that of its sub-sequences or super-sequences. Therefore, the pruning strategies of frequency-based SPM cannot be directly applied to UCSPM. Fourth, to identify HUCSPs, the utility values and positions of the candidate patterns must be stored during the mining process. Because numerous candidates exist, designing compact and accessible data structures for storing this rich information is important. 

Motivated by these challenges, in this paper, we propose a novel algorithm called FUCPM to mine HUCSPs more efficiently. The major contributions of our study are summarized as follows:

\begin{itemize}
	\item We formalized the problem of UCSPM and then proposed an efficient and scalable algorithm called FUCPM. Unlike the previous UCSPM algorithm, FUCPM can further handle multi-items-based sequence data.
	
	\item We designed two compact data structures called sequence information list (SIL) and instance chain to facilitate the utility calculation of candidate patterns.
	
	\item We proposed a novel utility upper bound called item-extension utilization (\textit{IEU}) and two search space pruning strategies: global unpromising items pruning (GUIP) and local unpromising items pruning (LUIP).
	
	\item After conducting substantial experiments on real-world and synthetic datasets, FUCPM shows its superiority in terms of runtime, memory consumption, and unpromising candidates pruning compared to the state-of-the-art algorithms. In addition, it demonstrates high scalability in handling large-scale sequence datasets (either single-item-based or multi-items-based).
	
\end{itemize}

The remainder of this paper is organized as follows. Section \ref{sec:relatedwork} reviews related works. Section \ref{sec:preliminaries} presents the basic preliminaries and problem statements. The proposed FUCPM algorithm with several data structures and pruning strategies is discussed in Section \ref{sec:method}. The experimental results and analysis are presented in Section \ref{sec:experiments}. Finally, Section \ref{sec:conclusion} concludes the paper and provides prospects for future work.

\section{Related work}  \label{sec:relatedwork}

In this section, we review prior works from the fields of SPM, HUSPM, CSPM, and UCSPM.

\subsection{SPM and HUSPM} \label{sec:SPM} 

Frequent itemset mining (FIM) \cite{agrawal1994fast} aims to discover frequent itemsets from transaction databases. Compared to FIM, SPM additionally considers the sequential ordering of items in the sequences; thus, it is a more complicated task. In general, SPM algorithms can be categorized into three types: breadth-first, depth-first, and pattern-growth algorithms \cite{fournier2017surveys}. Inspired by the Apriori algorithm \cite{agrawal1994fast} for FIM, Agrawal \textit{et al}. proposed the first SPM algorithm called AprioriAll \cite{agrawal1995mining} and GSP \cite{srikant1996mining}. Both AprioriAll and GSP are breadth-first; that is, they generate $k$-sequences (i.e., sequences with $k$ items) based on ($k$-1)-sequences and scan the database repeatedly to calculate the supports of all $k$-sequences. The algorithms do not terminate until no candidate sequences can be generated. The breadth-first algorithms suffer from multiple database scans, which are time-consuming. To address this problem, a depth-first algorithm, SPADE \cite{zaki2001spade}, is proposed. By utilizing a vertical database representation called IDList, which facilitates the calculation of support, SPADE avoids scanning the original database repeatedly. Due to the high efficiency of IDList, it has been used in several subsequent depth-first algorithms (e.g., SPAM \cite{ayres2002sequential}, CM-SPAM \cite{fournier2014fast}, and CM-SPADE \cite{fournier2014fast}). Although depth-first algorithms are more efficient, they generate many candidate patterns that do not exist in the database. To reduce the search space, Pei \textit{et al}. \cite{pei2004mining} proposed a pattern-growth algorithm called PrefixSpan, which only considers candidate patterns appearing in the database. PrefixSpan introduced an important concept, called the projected database, to reduce the cost of database scans. More detailed overviews of SPM can be found in \cite{fournier2017surveys, gan2019survey} .

Frequency-based SPM has its limitations because frequency cannot reflect importance under many circumstances. To address this problem, Ahmed \textit{et al}. \cite{ahmed2010novel} first incorporated the concept of utility into SPM and proposed two different algorithms called utility-level (UL) and utility-span (US), and both of them adopt an upper bound on utility called sequence-weighted utilization (\textit{SWU}) to prune the search space. UL and US have two main disadvantages. First, \textit{SWU} is quite loose, so the algorithms still generate numerous candidates. Second, it costs much time to calculate the exact utility of candidates because multiple database scans are required. In view of this, many later HUSPM algorithms \cite{yin2012uspan, alkan2015crom, wang2016efficiently, gan2020proum, gan2020fast} have focused on proposing tighter upper bounds to further prune the search space and design efficient data structures to facilitate the calculation of utility and upper bound value. Among them, USpan \cite{yin2012uspan} adopts the utility matrix to store the utility information of quantitative sequences. Two upper bounds (\textit{SWU} and sequence-projected utilization (\textit{SPU})) are utilized to eliminate the candidates. However, USpan cannot mine the complete set of HUSPs because \textit{SPU} can sometimes be less than utility, resulting in some HUSPs being filtered out. Based on the framework of USpan, Alkan \textit{et al}.  \cite{alkan2015crom} proposed a complete algorithm called HuspExt. The HuspExt algorithm uses a tight upper bound called Cumulate Rest of Match to reduce unpromising candidates. HUS-Span \cite{wang2016efficiently} is another famous HUSPM algorithm that utilizes two novel upper bounds (i.e., prefix extension utility and reduced sequence utility). In HUS-Span, a projected database called utility chain is designed to store the utility information of candidates. Recently, Gan \textit{et al}. proposed ProUM \cite{gan2020proum} and HUSP-ULL \cite{gan2020fast}, which use efficient data structures called utility array and utility-linked-list, respectively. A more comprehensive overview of HUSPM can be found in the latest literature reviews \cite{truong2019survey, gan2021survey}.

\subsection{CSPM and UCSPM} \label{sec:CSPM}

The methods mentioned in subsection \ref{sec:SPM} aim to find the complete set of sequential patterns; however, not all discovered patterns are useful under certain circumstances. In addition, the volume of mining results can be quite large, making it difficult for users to analyze and utilize them. To address this problem, several constraints (such as contiguous \cite{chen2007mining}, top-\textit{k} \cite{zhang2021tkus}, and closed \cite{ceci2021closed} constraints) have been proposed to obtain a concise and more meaningful subset of the complete set of sequential patterns. Among them, the contiguous constraint is significant when the underlying ordering of items in the sequence data must be emphasized.

Pan \textit{et al}. \cite{pan2005efficient} are among the earliest researchers to investigate the problem of CSPM. They developed two algorithms called MacosFSpan and MacosVSpan, which are inspired by the PrefixSpan algorithm, to mine frequent concatenate sequences (i.e., CSPs) from biological datasets. Kang \textit{et al}. \cite{kang2007mining} proposed the fixed-length spanning tree structure that is used to store the fixed-length contiguous sub-sequences and their supports of given sequence data. The formal definition of CSP was provided by Chen \textit{et al}. \cite{chen2007mining}. With the goal of discovering CSPs from web access logs, they designed the UpDown tree to store sequences that contain a given item, which benefits the mining process. Chen \textit{et al}. \cite{chen2009two} latter proposed an improved version of the UpDown tree, which can handle the sequence data where each itemset may contain more than one item. Currently, CSPM remains a research hotspot. Several studies \cite{zhang2015ccspan, abboud2017ccpm} combined closed and contiguous constraints for SPM with the aim of obtaining a more concise pattern set without loss of information. Other recent studies applied CSPM to a wide range of realistic applications, such as vehicle trajectory analysis \cite{bermingham2020mining, yang2018mining}, protocol specification extraction \cite{goo2019protocol}, and biological sequence analysis \cite{zhang2015mining, nawaz2021using}.

Frequency-based CSPM has been extensively studied; however, very few studies focused on utility-driven CSPM. Zhou \textit{et al}. \cite{zhou2007utility} first introduced utility into CSPM and developed a two-phase algorithm to mine web path traversal patterns, which are HUCSPs actually. This algorithm is similar to the UL algorithm mentioned previously; both generate candidate patterns in a breadth-first fashion and calculate the \textit{SWU} of these candidates in the first phase. In the second phase, the algorithm scans the database repeatedly to calculate the exact utility value of candidates with high \textit{SWU}. An improved algorithm for web path traversal pattern mining, called EUWPTM \cite{ahmed2009efficient}, was proposed later. To reduce a large number of candidates, EUWPTM uses a pattern-growth mechanism to generate candidate patterns. These two algorithms mainly adopt the loose upper bound \textit{SWU} to prune the search space. Huang \textit{et al}. \cite{huang2016efficient} proposed a more efficient algorithm HUCP-Miner with a tight upper bound called remaining utility upper-bound (\textit{RUUB}) that is actually equivalent to the prefix extension utility in HUS-Span. In addition, they designed a data structure called the UL-list to store the utility and position information of candidate patterns. Thus far, HUCP-Miner is the state-of-the-art algorithm for UCSPM. However, \textit{RUUB} is not tight enough to filter out unpromising patterns during the early stage of the mining process. In addition, all the aforementioned algorithms can only handle single-item-based sequences, whereas a large volume of multi-items-based sequences is produced and needs to be analyzed. These problems motivated us to develop a more efficient and scalable algorithm for the UCSPM task.

\section{Preliminaries}   \label{sec:preliminaries}
This section introduces the basic concepts and definitions used in this paper. Then, the formal problem statement of the UCSPM is provided.

\subsection{Concepts and Definitions}
Let $I$ = \{$i_{1}$, $i_{2}$, $\cdots$, $i_{n}$\} be a set of distinct items appearing in the database. An itemset $X$ is a nonempty subset of $I$, that is, $X \subseteq I$. The size of $X$ is defined as the number of items it contains and is denoted by $|X|$. Without loss of generality, the items contained in an itemset are arranged in a lexicographical order hereinafter. A sequence $S$ = $<$$X_{1}$, $X_{2}$, $\cdots$, $X_{m}$$>$ is an ordered list of itemsets, where $X_{k}$ $\subseteq$ $I$ for $1 \leq k \leq m$. We define $|S|$ = $\sum_{k = 1}^{m}|X_{k}|$ as the length of $S$. A sequence is called an $l$-sequence if its length is $l$. For example, given a set $I$ = \{$a, b, c, d, e, f$\}, $X$ = \{$cef$\} is an itemset with a size of three, and $S$ = $<$$\{a\}$, $\{bcf\}$, $\{ab\}$$>$ is a 6-sequence as it contains six items.

\begin{definition}[contiguous sub-sequence and super-sequence]	
	A sequence $S$ = $<$$X_{1}$, $X_{2}$, $\cdots$, $X_{m}$$>$ is a contiguous sub-sequence of another sequence $S'$ = $<$$X'_{1}$, $X'_{2}$, $\cdots$, $X'_{n}$$>$, $m \leq n$, which is denoted by $S \subseteq S'$, if there exists an integer $k$, $1 \leq k \leq n-m+1$, such that $X_{1}$ $\subseteq$ $X'_{k}$, $X_{2}$ $\subseteq$ $X'_{k+1}$, $\cdots$, $X_{m}$ $\subseteq $ $X'_{k+m-1}$. In addition, $S'$ is called a super-sequence of $S$.
\end{definition}

For example, given three sequences $S$ = $<$$\{a\}$, $\{af\}$$>$, $S'$ = $<$$\{e\}$, $\{ab\}$$>$, and $S''$ = $<$$\{c\}$, $\{ab\}$, $\{aef\}$$>$, we say $S$ is a contiguous sub-sequence of $S''$, while $S'$ is not a contiguous sub-sequence of $S''$.

\begin{definition}[quantitative sequence database]
	The database processed by UCSPM is a quantitative sequence database (abbreviated as  $q$-sequence database), which is described as follows: In a $q$-sequence database, each item is called a quantitative item ($q$-item). A $q$-item is represented as a tuple ($i$:$q$), where $i \in I$ and $q$ is the internal utility of item $i$. In addition, each item in $I$ has an external utility, which is listed separately. A quantitative itemset ($q$-itemset) $X$ with $m$ $q$-items is denoted by $X$ = \{($i_{1}$:$q_{1}$) ($i_{2}$:$q_{2}$) $\cdots$ ($i_{m}$:$q_{m})$\}. A $q$-sequence ($q$-sequence) $S$ = $<$$X_{1}$, $X_{2}$, $\cdots$, $X_{n}$$>$ is an ordered list of $n$ $q$-itemsets. Further, a $q$-sequence database is composed of a series of $q$-sequences with a unique identifier \textit{SID}. 
\end{definition}

Table \ref{table1} shows a $q$-sequence database with five $q$-sequences and six distinct items. In this study, we use this $q$-sequence database as the running example. Table \ref{table2} lists the external utility of each item. In $S_{1}$, the first $q$-item ($b$:$2$) represents that item $b$ has an internal utility of $2$. The first itemset of $S_{1}$ (i.e., \{($b$:$2$) ($f$:$4$)\}) is a $q$-itemset containing two $q$-items: ($b$:$2$) and ($f$:$4$). $S_{1}$ = $<$\{($b$:$2$) ($f$:$4$)\}$, $\{($a$:$2$) ($e$:$2$)\}$, $\{($c$:$2$) ($e$:$1$)\}$>$ is a $q$-sequence containing three $q$-itemsets: \{($b$:$2$) ($f$:$4$)\}, \{($a$:$2$) ($e$:$2$)\}, and \{($c$:$2$) ($e$:$1$)\}.

\begin{table}[h]
	\centering
	\caption{Running example of a $q$-sequence database}
	\renewcommand{\arraystretch}{1.25}
	\label{table1}
	\begin{tabular}{|c|c|}  
		\hline 
		\textbf{SID} & \textbf{$q$-sequence} \\
		\hline  
		\({S}_{1}\) & $<$\{($b$:2) ($f$:4)\}$,$ \{($a$:2) ($e$:2)\}$,$ \{($c$:2) ($e$:1)\}$>$ \\ 
		\hline
		\({S}_{2}\) & 
		$<$\{($a$:1)\}$,$ \{($c$:2) ($d$:1)\}$,$ 
		\{($a$:1) ($b$:1) ($e$:2)\}$>$ \\
		\hline 
		\({S}_{3}\) & $<$\{($b$:2) ($f$:2)\}$,$ \{($f$:2)\}$,$ 
		\{($a$:3) ($d$:1)\}$>$ \\
		\hline  
		\({S}_{4}\) & 
		$<$\{($d$:1)\}$,$ 
		\{($b$:4) ($f$:5)\}$,$
		\{($c$:1) ($e$:2)\}$,$
		\{($f$:1)\}$>$ \\
		\hline
		\({S}_{5}\) & 
		$<$\{($a$:2)\}$,$
		\{($a$:1) ($c$:3)\}$,$ \{($c$:1) ($f$:2)\}$,$
		\{($b$:1)\}$>$ \\
		\hline
	\end{tabular}
\end{table}

\begin{table}[h]
	\caption{External utility table}
	\renewcommand{\arraystretch}{1.25}
	\label{table2}
	\centering
	\begin{tabular}{|c|c|c|c|c|c|c|}
		\hline
		\textbf{Item}	    & $a$	& $b$	& $c$	& $d$	& $e$	& $f$ \\ \hline 
		\textbf{External utility}	& 3 & 2& 3 & 2 & 1 & 1\\ \hline
	\end{tabular}
\end{table}

\begin{definition}[matching]
	Given an itemset $X$ = \{$i_{1}$, $i_{2}$, $\cdots$, $i_{m}$\} and a $q$-itemset $Y$ = \{($j_{1}$:$q_{1}$) ($j_{2}$:$q_{2}$) $\cdots$ ($j_{m}$:$q_{m}$)\}, we say that $X$ is the matching of $Y$, denoted by $X \sim Y$, if and only if $i_{k}$ = $j_{k}$ for $1 \leq k \leq m$. Similarly, given a sequence $S$ = $<$$X_{1}$, $X_{2}$, $\cdots$, $X_{n}$$>$, and a $q$-sequence $Q$ = $<$$Y_{1}$, $Y_{2}$, $\cdots$, $Y_{n}$$>$, we say that $S$ is the matching of $Q$, denoted by $S \sim Q$, if and only if $X_{k}\sim Y_{k}$ for $1 \leq k \leq n$. 
\end{definition}

\begin{definition}[instance]
	Given a sequence $S$ = $<$$X_{1}$, $X_{2}$, $\cdots$, $X_{m}$$>$ and a $q$-sequence $Q$ = $<$$Y_{1}$, $Y_{2}$, $\cdots$, $Y_{n}$$>$, where $m \leq n$. If there exists an integer $p$, $m \leq p \leq n$, such that $X'_{k}$ $\sim Y_{p - m + k}$ $\land X_{k} $ $ \subseteq X'_{k} $ for $1 \leq k \leq m$ (i.e., $S$ is the contiguous sub-sequence of the matching of $Q$), we say that $Q$ has an instance of $S$ at \textit{ending position} $p$. It should be noted that $Q$ may have several instances of $S$, which correspond to different ending positions. We denote the set of these ending positions as \textit{EP}($S, Q$). In addition, we say that $Q$ contains $S$ if $Q$ has at least one instance of $S$, which can be denoted by $S' \sim Q$ $\land S $ $\subseteq S'$. In the following, we use $S \sqsubseteq Q$ to represent that $Q$ contains $S$ for convenience.
\end{definition}

For example, itemset $\{bf\}$ is the matching of \{($b$:$2$) ($f$:$4$)\}. Sequence $<$$\{bf\}$, $\{ae\}$, $\{ce\}$$>$ is the matching of $S_{1}$. Given a sequence $S$ = $<$$\{a\}$, $\{c\}$$>$, $S_{5}$ has instances of $S$ at ending positions 2 and 3, respectively. Therefore, \textit{EP}($S,\ S_{5}$) = $\{2, 3\}$. In addition, we say that $S_{5}$ contains $S$ and denote it as $S \sqsubseteq S_{5}$. 

Subsequently, we define the calculation methods for utility values in different situations. Note that the internal utility of item $i$ within the $j$-th $q$-itemset in the $q$-sequence $Q$ is denoted by $q$($i,\ j,\ Q$), and the external utility of $i$ is denoted by $p$($i$).

\begin{definition}[utility calculation]
	Given a $q$-sequence $Q$, the utility of the $q$-item $i$ in the $j$-th $q$-itemset in $Q$, which is denoted by $u$($i, j, Q$), can be calculated as $u$($i, j, Q$) = $q$($i, j, Q$) $\times$ $p$($i$). In addition, the utility of a $q$-itemset or a $q$-sequence is defined as the sum of utility values of elements (i.e., $q$-items and $q$-itemsets, respectively) it contains. Furthermore, we formalize the utility calculation of instances as follows: Given an itemset $X$, a sequence $S$ = $<$$X_{1}$, $X_{2}$, $\cdots$, $X_{m}$$>$, and a $q$-sequence database $D$, the utility of $X$ in the $j$-th $q$-itemset in $Q$ is defined as $u$($X, j, Q$) = $\sum_{i \in X}^{}u$($i,\ j,\ Q$). Assuming that $Q$ has an instance of $S$ at the ending position $p$, the utility of this instance can be calculated as $u$($S,\ p,\ Q$) = $\sum_{j=1}^{m}u$($X_j,\ p - m + j,\ Q$). The utility of $S$ in $Q$ is defined as the maximal utility among all instances of $S$ in $Q$; that is, $u$($S, Q$) = $\max\{u$($S,\ p,\ Q$)$ | \forall p \in$ \textit{EP}($S, Q$)\}. Finally, the utility of $S$ in $D$ is defined as $u$($S$) = $\sum_{Q \in D}^{}u$($S, Q$).
\end{definition}

For example, the utility of ($b$:$2$) within the first $q$-itemset of $S_{1}$ is calculated as $u$($b,\ 1,\ S_{1}$) = $q$($b,\ 1,\ S_{1}$) $\times$ $p$($b$) = 2 $\times$ 2 = 4. Meanwhile, the utility of the first $q$-itemset of $S_{1}$ is $u$($X_{1}, S_{1}$) = 4 + 4 = 8, and that of $S_{1}$ is $u$($S_{1}$) = 8 + 8 + 7 = 23.

In addition, the utility of $\{ab\}$ in the third $q$-itemset of $S_{2}$ is $u$($\{ab\}$, 3, $S_{2}$) = 3 + 2 = 5. Given a sequence $S$ = $<$$\{a\}$, $\{c\}$$>$, $u$($S, 2, S_{5}$) = 6 + 9 = 15 and $u$($S, 3, S_{5}$) = 3 + 3 = 6 can be obtained. Therefore, $u$($S$, $S_{5}$) = max$\{15, 6\}$ = 15. Finally, considering the $q$-sequence database in Table \ref{table1}, we have $u$($S$) = 12 + 9 + 15 = 36.

\subsection{Problem Statement} \label{Problem statement}
\begin{definition}[high-utility contiguous sequential pattern] \label{hucsp}
	In a $q$-sequence database $D$, a sequence $S$ is called a HUCSP, if it is the contiguous sub-sequence of some sequences in $D$ and satisfies that $u$($S$) $\ge$ $\xi $ $\times $ $u(D)$, where $\xi$ is the minimum utility threshold that is given as a percentage. 
\end{definition}

\textbf{Problem Statement} Given a $q$-sequence database $D$, an external utility table and a minimum utility threshold $\xi$, the problem of UCSPM is to identify the complete set of HUCSPs in $D$.

To illustrate the problem of UCSPM clearly, an example is given as follows: Considering the $q$-sequence database given in Table \ref{table1}, we can obtain $u$($D$) = 106. When $\xi$ = 25\%, the minimum utility threshold value is 25\% $\times$ 106 = 26.5, and the discovered HUCSPs are $<$\{$bf$\}$>$ and $<$$\{a\}$, $\{c\}$$>$ with a utility of 27 and 36, respectively. Evidently, $<$\{$bf$\}$>$ is a contiguous sub-sequence of the matching of $S_{1}$, $S_{3}$, $S_{4}$, and $<$$\{a\}$, $\{c\}$$>$ is a contiguous sub-sequence of the matching of $S_{1}$, $S_{2}$, and $S_{5}$.

\section{Proposed Method}   \label{sec:method}

This section describes the proposed FUCPM algorithm for mining HUCSPs in a pattern-growth manner. FUCPM utilizes the SIL and instance chain data structures to avoid scanning the $q$-sequence database repeatedly when calculating the utility of candidate patterns. Two powerful pruning strategies (i.e., GUIP and LUIP) are used to reduce the search space. The details of the above data structures, pruning strategies, and the FUCPM algorithm are provided in the following section. To facilitate the discussion in this section, we first define the following. 

\begin{definition}[extension]
	The extension operation is used to generate new candidate patterns and consists of two types of operations: item-extension (I-extension) and sequence-extension (S-extension). Given a sequence $S$ and an item $i$, the I-extension sequence of $S$, denoted by $<$$ S \bigoplus i $$>$, appends $i$ to the last itemset of $S$. The S-extension sequence of $S$, denoted by $<$$ S \bigotimes i $$>$, adds $i$ to a new itemset and then appends the new itemset to the end of $S$. For brevity, the I-extension sequence and S-extension sequence are collectively referred to as extension sequences hereinafter. Note that the extension sequences of $S$ are a subset of the super-sequences of $S$.
\end{definition}

For example, given a sequence $S$ = $<$$\{a\}$, $\{c\}$$>$, and an item $d$, then $<$$\{a\}$, $\{cd\}$$>$ is the I-extension sequence of $S$, and $<$$\{a\}$, $\{c\}$, $\{d\}$$>$ is the S-extension sequence of $S$.

\begin{definition}[extension item]
	\label{extension item}
	Given a sequence $S$ whose last item is $i$, a $q$-sequence $Q$, and a $q$-sequence database $D$, we assume that $Q$ has $n$ instances of $S$ and the set of ending positions is \textit{EP}($S, Q$) = $\{ep_{1}$, $ep_{2}$, $\cdots$, $ep_{n}\}$. The set of I-extension items of $S$ in $Q$, denoted by \textit{Iitem}($S, Q$), comprises the items that appear in the $ep_{1}$, $ep_{2}$, $\cdots$, $ep_{n}$-th itemset of $Q$ and are lexicographically larger than $i$. The set of S-extension items of $S$ in $Q$, denoted by \textit{Sitem}($S, Q$), is composed of the items appearing in the ($ep_{1}$ + 1), ($ep_{2}$ + 1), $\cdots$, ($ep_{n}$ + 1)-th itemset of $Q$. Furthermore, the set of I-extension/S-extension items of $S$ in $D$ is defined as \textit{Iitem}($S$) = $\bigcup_{{Q \in D}}^{}$\textit{Iitem}($S, Q$) and \textit{Sitem}($S$) = $\bigcup_{{Q \in D}}^{}$\textit{Sitem}($S, Q$), respectively.
\end{definition}	

For example, the sets of I-extension/S-extension items of $<$\{$a$\}$>$ in $S_{2}$ are \textit{Iitem}($<$\{$a$\}$>$, $S_{2}$) = \{$b, e$\}, and \textit{Sitem}($<$\{$a$\}$>$, $S_{2}$) = \{$c, d$\}. Furthermore, given the database shown in \ref{table1}, we have \textit{Iitem}($<$\{$a$\}$>$) = \{$b, c, d, e$\} and \textit{Sitem}($<$\{$a$\}$>$) = \{$a, c, d, e, f$\}.

\begin{definition}[remaining sequence and remaining utility]
	Given a sequence $S$ and a $q$-sequence $Q$, it is assumed that $Q$ has an instance of $S$ at the ending position $p$. The remaining sequence of $Q$ with respect to such an instance is denoted by $Q/_{(S,p)}$ and is defined as a suffix sequence of $Q$, which begins from the item after the last item of such an instance in $Q$ to the end of $Q$. Furthermore, the utility of the remaining sequence is called the remaining utility and is defined as \textit{ru}($Q/_{(S,p)}$) = $\sum_{i \in Q/_{(S,p)}}^{}u$($i$).
\end{definition}

For example, the remaining sequence of the instance of $<$$\{a\}$, $\{c\}$$>$ at ending position 2 in $S_{5}$ is $Q/_{(<\{a\}, \{c\}>, 2)}$ = $<$\{($c$:$1$) ($f$:$2$)\}$,$ \{($b$:$1$)\}$>$, and the corresponding remaining utility is $u(S_{5}/_{(<\{a\}, \{c\}>, 2)})$ = 3 + 2 + 2 = 7. 

\subsection{Data Structures}
As mentioned in the introduction, the main challenge of UCSPM is the existence of numerous candidate patterns. In addition, each candidate pattern may appear multiple times in a $q$-sequence. Therefore, to calculate the utility of a candidate pattern, the whole $q$-sequence database should be scaned to find all its instances in each $q$-sequence. Clearly, the scanning process requires a long execution time. To address this problem, we proposed using SIL to represent the original database and using instance-chain (IChain) to store the utility and positions of instances of a candidate pattern. The details of SIL and IChain are described below.

SIL stores information of the $q$-sequence, including the utility of each $q$-item and the remaining utility related to each $q$-item. Table \ref{SIL} shows SIL of $S_{1}$ in Table \ref{table1}. In SIL of $S_{1}$, each curly bracket ($\{\}$) corresponds to a $q$-itemset. The element ($b,$ 4$,$ 19) within the first curly bracket indicates that the utility of item $b$ is four, and the remaining utility of $S_{1}$ with respect to this item is 19. The storage of utility and remaining utility in the SIL benefits the calculation of utility and upper bound value of candidate patterns, which will be discussed later. 

IChain stores the utility and ending positions of all instances of a candidate pattern. To facilitate the following description, the $q$-sequence $Q$ is assumed to have $n$ instances of candidate pattern $S$. The set of these ending positions is \textit{EP}($S, Q$) = $\{ep_{1}, ep_{2}, \cdots, ep_{n}\}$. IChain comprises several instance lists, each of which corresponds to a $q$-sequence containing the candidate pattern. For the example of $S$ and $Q$, the instance list contains the SID of $Q$, as well as a list of $n$ elements. The $i$-th element contains two fields: (1) \textit{EPos}, which is the ending position of the $i$-th instance of $S$ in $Q$ (i.e., $ep_{i}$) and (2) \textit{Utility}, which is the utility of the $i$-th instance of $S$ in $Q$. In the following, we use a tuple (\textit{EPos}, \textit{Utility}) to represent an element in instance list. The IChain of all 1-sequences can be constructed by scanning the SIL once. Figure \ref{chainofa} shows the IChain of sequence $<$$\{a\}$$>$ in the $q$-sequence database given in Table \ref{table1}.

\begin{table}[h]
	\caption{Sequence information list of $S_{1}$}
	\renewcommand{\arraystretch}{1.25}
	\newcommand{\tabincell}[2]{\begin{tabular}{@{}#1@{}}#2\end{tabular}} %
	\label{SIL}
	\centering
	\begin{tabular}{|c|c|}
		\hline
		\textbf{SID}	    & Sequence information list	\\ \hline 
		\textbf{$S_{1}$}	& \tabincell{c}{$<$\{($b,$ 4$,$ 19) ($f,$ 4$,$ 15)\}$,$ \{($a,$ 6$,$ 9) ($e,$ 2$,$ 7)\}$,$ \\\{($c,$ 6$,$ 1) ($e,$ 1$,$ 0)\}$>$} \\ \hline
	\end{tabular}
\end{table}

\begin{figure}[h]
	\centering
	\includegraphics[width=0.5\linewidth]{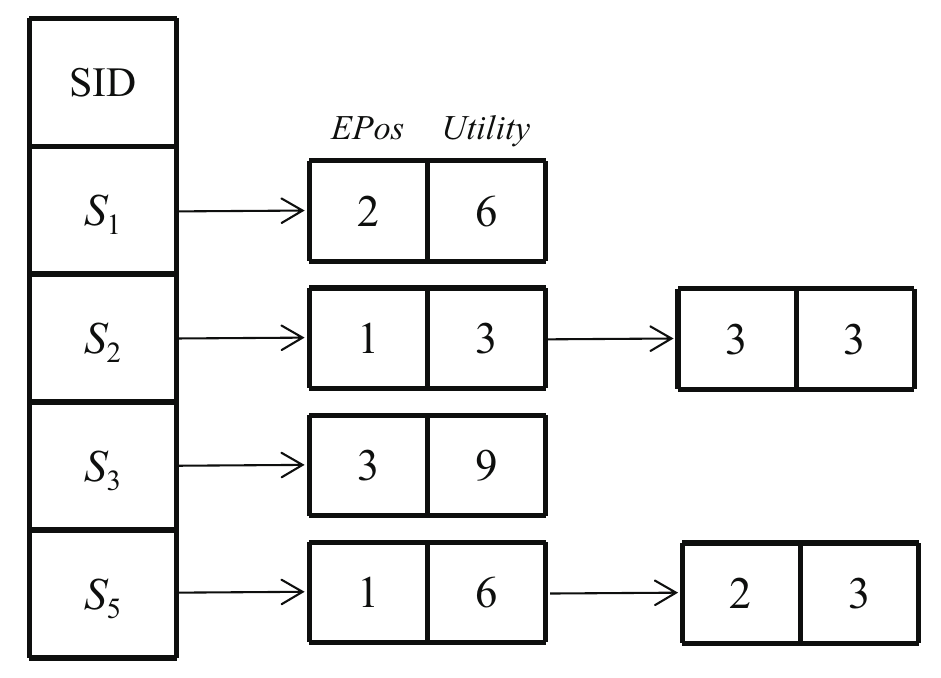}
	\caption{The IChain of $<$${a}$$>$}
	\label{chainofa}
\end{figure}

The IChain for $k$-sequences, where $k > 1$, is constructed recursively based on the IChain of its prefix with a length of $k$-1. Specifically, given sequences $S$ and $S'$ = $<$$ S \bigoplus i' $$>$, to construct the IChain of $S'$, we must scan the IChain of $S$ using the following method. For a certain element (\textit{EPos}, \textit{Utility}) in the instance list of $S$ corresponding to a $q$-sequence $Q$, we check the SIL of $Q$ whether item $i'$ exists in the $\textit{EPos}$-th itemset. If there exists, we construct a new element (\textit{EPos}, \textit{Utility} + $u$($i'$, \textit{EPos}, $Q$)) and insert it into the instance list of $S'$ corresponding to $Q$. After traversing all elements of the IChain of $S$ in the same manner, the IChain of $S'$ can be built. For an S-extension sequence $S''$ = $<$$ S \bigotimes i'' $$>$, the construction method of the IChain of $S''$ is similar to that of $S'$. The only difference is that we check whether $i''$ exists in the (\textit{EPos} + 1)-th itemset of $Q$. If there exists, we construct an element (\textit{EPos} + 1, \textit{Utility} + $u$($i''$, \textit{EPos} + 1, $Q$)) and insert it to the instance list of $S''$. Figure \ref{chainofac} shows the IChain of sequence $<$$\{a\}$, $\{c\}$$>$ in the $q$-sequence database given in Table \ref{table1}. 

\begin{figure}[h]
	\centering
	\includegraphics[width=0.5\linewidth]{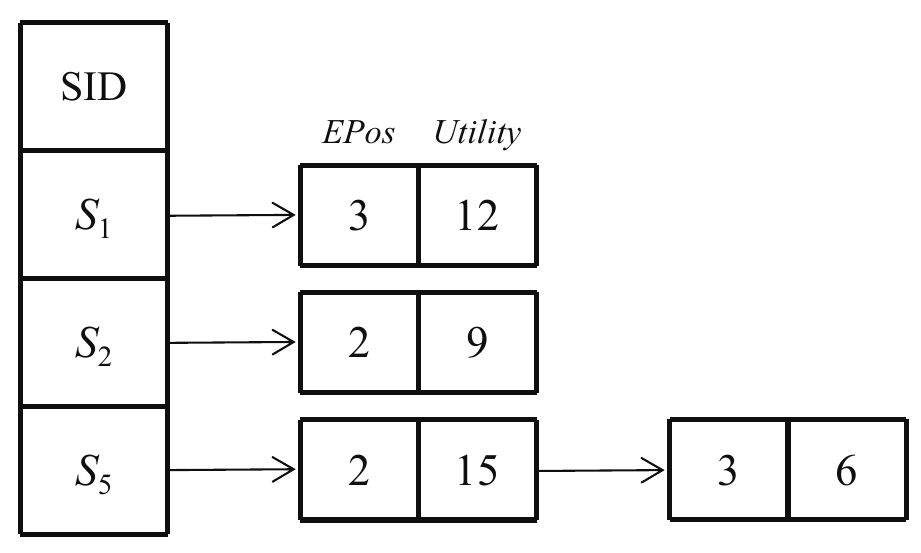}
	\caption{The IChain of $<$$\{a\}$, $\{c\}$$>$}
	\label{chainofac}
\end{figure}

\subsection{Pruning Strategies}
In the SPM and HUSPM tasks, the combinatorial explosion of the search space is a classic problem, which is caused by the sequential ordering of itemsets. Many SPM algorithms adopt the downward closure property to prune the search space. This property states that, if a sequence is infrequent, then all its super-sequences are also infrequent. However, this property only holds for the frequency, but does not hold for utility. To address this problem, researchers have proposed using upper bounds on utility to prune the search space in HUSPM. An upper bound on utility must satisfy the following two properties: (1) overestimate property (i.e., the upper bound value of a sequence overestimates the utility of this sequence) and (2) the downward closure property (i.e., the upper bound value of a sequence is no less than that of its extension sequences). Sequence-weighted utilization (\textit{SWU}) \cite{yin2012uspan} is a commonly used upper bound in HUSPM. Inspired by \textit{SWU}, we develop a GUIP strategy to prune low-utility 1-sequences (also called global unpromising items) in the early stage of mining HUCSPs. Furthermore, we proposed a novel upper bound called item-extension utilization (\textit{IEU}) that considers the contiguous constraint and a corresponding pruning strategy to eliminate local unpromising items for each candidate during the entire mining process. Subsequently, we introduce the details of these two upper bounds and the corresponding pruning strategies. 

\begin{definition}[sequence-weighted utilization]
	Given a sequence $S$ and a $q$-sequence database $D$, the sequence-weighted utilization (\textit{SWU}) of $S$ in $D$ is defined as 
	\textit{SWU}($S$) = $\sum_{S \sqsubseteq Q \land Q \in D }^{}\textit{u}(Q)$.
\end{definition} 

For example, in Table \ref{table1}, \textit{SWU}($a$) =  $u$($S_{1}$) +  $u$($S_{2}$) +  $u$($S_{3}$) +  $u$($S_{5}$) = 23 + 18 + 19 + 25 = 85, and \textit{SWU}($d$) = $u$($S_{2}$) +  $u$($S_{3}$) +  $u$($S_{4}$) = 58.

\begin{theorem}[overestimate property of \textit{SWU}]
	\label{OPSWU}
	Given a sequence $S$ and a $q$-sequence database $D$, it can be obtained that $u(S) \leq \textit{SWU}(S)$.
\end{theorem}

\begin{proof}
	\label{proof_OPSWU}
	For any $q$-sequence $Q$, we can obtain that \textit{u}($S, Q$) $\leq$ $u$($Q$). Therefore, $u$($S$) = $\sum_{S \sqsubseteq Q \land Q \in D}^{}u$($S, Q$) $\leq \sum_{S \sqsubseteq Q\land Q \in D}^{}u$($Q$) = \textit{SWU}($S$). 
\end{proof}

\begin{theorem}[downward closure property of \textit{SWU}]
	\label{DCPSWU}
	Given two sequences $S$ and $S'$, and a $q$-sequence database $D$, if $S'$ is a super-sequence of $S$ (i.e., $S \subseteq S'$), then \textit{SWU}($S'$) $\leq$ \textit{SWU}($S$).
\end{theorem}

\begin{proof}	
	\label{proof_DCPSWU}
	Because $S \subseteq S'$, we have $\{Q|S'\sqsubseteq Q \land Q \in D\}$ $\subseteq$ $\{Q|S\sqsubseteq Q \land Q \in D\}$. Therefore, we can obtain that \textit{SWU}($S'$) = $\sum_{S' \sqsubseteq Q \land Q \in D }^{}\textit{u}$($Q$) $\leq \sum_{S \sqsubseteq Q \land Q \in D}^{}\textit{u}$($Q$) = \textit{SWU}($S$).
\end{proof}

Some existing HUSPM algorithms \cite{gan2020proum, gan2020fast} utilize \textit{SWU} to prune low-utility 1-sequence. That is, if the \textit{SWU} value of a 1-sequence is less than the minimum utility threshold, then this 1-sequence and all its super-sequences cannot be HUSPs. In other words, the unique item in the low-utility 1-sequence is a global unpromising item, and the sequences containing this item cannot be HUSPs. Numerous unpromising candidate patterns can be pruned using the \textit{SWU}. However, existing algorithms only use the strategy once, which cannot eliminate global unpromising items thoroughly. Hence, we propose a more effective pruning strategy called GUIP.

\begin{strategy}[GUIP strategy]
	The GUIP strategy is a recurrent process: (1) for each item $i$ in the $q$-sequence database, if \textit{SWU}($i$) $< \xi$ $\times $ $u$($D$), then remove $i$ from the database; (2) update the utility of each $q$-sequence and the SWU of the remaining items; (3) go to (1) until, for any remaining item $i$, \textit{SWU}($i$) $\ge$ $\xi $ $\times$ $u$($D$).
\end{strategy} 

In addition to the GUIP strategy, we further propose the \textit{IEU} upper bound and the corresponding LUIP strategy to prune the local unpromising items of each candidate pattern.

\begin{definition}[item-extension utilization]
	Given a sequence $\alpha$ and a $q$-sequence $Q$ containing $\alpha$, assume that $S$ is the extension sequence of $\alpha$ where the extension item is $i$. The $p$-th itemset of $Q$ is denoted by $Q^{p}$.
	
	i) For I-extension (i.e., $S$ = $\alpha \bigoplus i$), the \textit{IEU} of $S$ in $Q$ with respect to the ending position $p$, $p \in$  \textit{EP}($\alpha, Q$), is defined as
	\[\textit{IEU}(S, p, Q) = \textit{u}(\alpha, p, Q) + \textit{u}(i, p, Q) + \textit{ru}(Q/_{(i,p)}),\]
	if and only if $i \in Q^{p}$; otherwise, \textit{IEU}($S, p, Q$) = 0.
	
	ii) For S-extension (i.e., $S$ = $\alpha \bigotimes i$), the \textit{IEU} of $S$ in $Q$ with respect to the ending position $p$, $p \in  \textit{EP}$($\alpha, Q$), is defined as
	\[\textit{IEU}(S, p, Q) = \textit{u}(\alpha, p, Q) + \textit{u}(i, p+1, Q) + \textit{ru}(Q/_{(i,p+1)}),\]
	if and only if $i \in Q^{p+1}$; otherwise, \textit{IEU}($S, p, Q$) = 0.
	
	iii) For both I-extension and S-extension, the \textit{IEU} of $S$ in $Q$ is defined as 
	\[\textit{IEU}(S, Q) = \mathop{max}\limits_{ p \in \textit{EP}(\alpha, Q)}{\textit{IEU}}(S, p, Q).\]
	The \textit{IEU} of $S$ in $q$-sequence database $D$ is defined as
	\[\textit{IEU}(S) = \sum_{S \sqsubseteq Q \land Q \in D}^{}\textit{IEU}(S, Q).\]
\end{definition}

For example, in Table \ref{table1}, consider the sequence $\alpha$ = $<${$a$}$>$, I-extension item $e$ and S-extension item $c$. Then, for the I-extension sequence $S$ = $<$$\{ae\}$$>$, we have \textit{IEU}($S, 2, S_{1}$) = 6 + 2 + 7 = 15, \textit{IEU}($S, 2, S_{2}$) = 3 + 2 + 0 = 5. Finally, \textit{IEU}($S$) = \textit{IEU}($S, S_{1}$) + \textit{IEU}($S, S_{2}$) = 15 + 5 = 20. 

For the S-extension sequence $S'$ = $<$$\{a\}, \{c\}$$>$, we have \textit{IEU}($S', 2, S_{1}$) = 6 + 6 + 1 = 13, \textit{IEU}($S', 1, S_{2}$) = 3 + 6 + 9 = 18, \textit{IEU}($S', 1, S_{5}$) = 6 + 9 + 7 = 22 and \textit{IEU}($S', 2, S_{5}$) = 3 + 3 + 4 = 10. The \textit{IEU} value of $S'$ in $S_{5}$ is \textit{IEU}($S', S_{5}$) = max\{22, 10\} = 22. Finally, \textit{IEU}($S'$) = \textit{IEU}($S', S_{1}$) + \textit{IEU}($S', S_{2}$) + \textit{IEU}($S', S_{5}$) = 13 + 18 + 22 = 53.

\begin{theorem}[overestimate property of \textit{IEU}]
	\label{OPIEU}
	Given a sequence $S$ and a $q$-sequence database $D$, we can obtain that $u(S)\leq IEU(S)$.
\end{theorem}

\begin{proof} \label{proof_OPIEU}
	It is assumed that $S$ is the extension sequence of $\alpha$ where the extension item is $i$. 
	
	i) For I-extension, we have
	\begin{align*}
	\textit{u}(S, Q) &= \mathop{max}\limits_{ p \in \textit{EP}(\alpha, Q)}{\textit{u}}(S, p, Q) \\
	& = \mathop{max}\limits_{ p \in \textit{EP}(\alpha, Q)}{\textit{u}}(\alpha, p, Q) + {\textit{u}}(i, p, Q) \\
	&\leq \mathop{max}\limits_{ p \in \textit{EP}(\alpha, Q)}{\textit{u}}(\alpha, p, Q) + {\textit{u}}(i, p, Q) + {\textit{ru}}(Q/_{(i, p)}) \\
	& = \textit{IEU}(S, Q). 
	\end{align*}
	
	ii) For S-extension, we have
	\begin{align*}
	\textit{u}(S, Q) &= \mathop{max}\limits_{ p \in \textit{EP}(\alpha, Q)}{\textit{u}}(S, p+1, Q) \\
	& = \mathop{max}\limits_{ p \in \textit{EP}(\alpha, Q)}{\textit{u}}(\alpha, p, Q) + {\textit{u}}(i, p+1, Q) \\
	&\leq \mathop{max}\limits_{ p \in \textit{EP}(\alpha, Q)}{\textit{u}}(\alpha, p, Q)+{\textit{u}}(i, p+1, Q) + {\textit{ru}}(Q/_{(i, p+1)}) \\
	& = \textit{IEU}(S, Q). 
	\end{align*}
	
	iii) For both I-extension and S-extension, \textit{u}($S$) = $\sum_{S \sqsubseteq Q \land Q \in D}^{}\textit{u}$($S, Q$) $\leq \sum_{S \sqsubseteq Q\land Q \in D}^{}$\textit{IEU}($S, Q$) = \textit{IEU}($S$) can be obtained.
\end{proof}

\begin{theorem}[downward closure property of \textit{IEU}]
	\label{DCPIEU}
	Given two sequences $S$ and $S'$, and a $q$-sequence database $D$, if $S'$ is the extension sequence of $S$, then $\textit{IEU}(S')\leq \textit{IEU}(S)$ can be obtained.
\end{theorem}

\begin{proof}  \label{proof_DCPIEU}
	It is assumed that $S'$ is the extension sequence of $S$, where the extension item is $i'$, and $S$ is the extension sequence of $\alpha$, where the extension item is $i$. According to Definition \ref{extension item}, $i'$ is arranged after $i$ in any $q$-sequence $Q$ that contains both $S$ and $S'$. 
	
	i) First, if $S$ = $\alpha \bigoplus i$, for I-extension (i.e., $S’$ = $S \bigoplus i’$), we have
	\begin{align*}
	\textit{IEU}(S', Q) &= \mathop{max}\limits_{p' \in \textit{EP}(S, Q)}{\textit{u}}(S, p', Q) + {\textit{u}}(i', p', Q) + {\textit{ru}}(Q/_{(i', p')})\\
	&\leq \mathop{max}\limits_{ p' \in \textit{EP}(S, Q)}{\textit{u}}(S, p', Q) + {\textit{ru}}(Q/_{(i, p')}) \\
	& = \mathop{max}\limits_{ p \in \textit{EP}(\alpha, Q)}{\textit{u}}(\alpha, p, Q) + {\textit{u}}(i, p, Q) + {\textit{ru}}(Q/_{(i, p)}) \\
	& = \textit{IEU}(S, Q). 
	\end{align*}
	
	ii) For S-extension (i.e., $S'$ = $S \bigotimes i'$), we have
	\begin{align*}
	\textit{IEU}(S', Q) &= \mathop{max}\limits_{p' \in \textit{EP}(S, Q)}{\textit{u}}(S, p', Q) + {\textit{u}}(i', p'+1, Q) \\
	& \qquad  \qquad \quad \  + {\textit{ru}}(Q/_{(i', p'+1)})\\
	&\leq \mathop{max}\limits_{ p' \in \textit{EP}(S, Q)}{\textit{u}}(S, p', Q) + {\textit{ru}}(Q/_{(i, p')}) \\
	& = \mathop{max}\limits_{ p \in \textit{EP}(\alpha, Q)}{\textit{u}}(\alpha, p, Q) + {\textit{u}}(i, p, Q) + {\textit{ru}}(Q/_{(i, p)}) \\
	& = \textit{IEU}(S, Q). 
	\end{align*}
	
	Similarly, it can be proved that $\textit{IEU}$($S', Q$) $\leq$ \textit{IEU}($S, Q$) if $S$ = $\alpha \bigotimes i$. Because $S'$ is the extension sequence of $S$, we have $\{Q|S' \sqsubseteq Q \land Q \in D\} \subseteq \{Q|S \sqsubseteq Q\land Q \in D\}$. Therefore, \textit{IEU}($S'$) = $\sum_{S' \sqsubseteq Q \land Q \in D}^{}\textit{IEU}$($S', Q$) $\leq \sum_{S' \sqsubseteq Q\land Q \in D}^{}\textit{IEU}$($S, Q$) $\leq \sum_{S \sqsubseteq Q\land Q \in D}^{}$\textit{IEU}($S, Q$) = \textit{IEU}($S$).
\end{proof}

\begin{strategy}[LUIP strategy]
	Let $S$ be a candidate pattern; if $\textit{IEU}$($S$) $< \xi $ $\times$  $u$($D$), then $S$ and all its extension sequences can be pruned from the search space. That is, assume that $S$ is the extension sequence of sequence $\alpha$, where the extension item is $i$, if $\textit{IEU}$($S$) $< \xi$ $\times$ $u$($D$), then $i$ is a local unpromising item for $\alpha$ and can be pruned.
\end{strategy} 

\subsection{Proposed FUCPM Algorithm}
Based on the aforementioned data structures and pruning strategies, the proposed FUCPM algorithm is described as follows. Algorithm \ref{alg:FUCPM} shows the pseudocode of the main procedure of FUCPM, which takes a $q$-sequence database $D$, an external utility table \textit{EUT}, and a minimum utility threshold $\xi$ as the inputs. FUCPM first scans $D$ to calculate the utility of each $q$-sequence in $D$ and the utility of $D$. Further, FUCPM follows the GUIP strategy to delete global unpromising items and obtain the revised database $D'$ (line 1). Then, FUCPM scan $D'$ to construct the SIL of each $q$-sequence and IChain of each 1-sequence (line 2). The utility of each 1-sequence is obtained from its IChain, and FUCPM determines whether the 1-sequence is a HUCSP (lines 4--6). Subsequently, FUCPM calls the Recursive-Search procedure to discover longer HUCSPs recursively (line 7). Finally, the algorithm returns the complete set of HUCSPs (line 9).

\begin{algorithm}[ht]
	\small
	\caption{FUCPM algorithm}
	\label{alg:FUCPM}
	\begin{algorithmic}[1]
		\REQUIRE 
		$D$: a $q$-sequence database; \textit{EUT}: external utility table; $\xi$: minimum utility threshold.
		\ENSURE 
		\textit{HUCSPs}: the set of high-utility contiguous sequential patterns.
		
		\STATE scan $D$ to: \\
		(1) calculate $u$($Q$) for each $Q \in D$ and calculate $u$($D$); \\
		(2) calculate \textit{SWU}($S$) for each $S \in $ 1-sequences recurrently and obtain the revised database $D'$ by deleting the unique item contained in $S$ such that \textit{SWU}($S$) $< \xi$ $\times$ $u$($D$); $//$ \textit{The GUIP strategy}\\
		\STATE scan $D'$ to: \\
		(1) construct the SIL of each $q$-sequence; \\
		(2) construct the IChain of each 1-sequence;
		
		\FOR {each $S \in $ 1-sequences}
		\IF{$u$($S$) $\ge \xi $ $\times $ $u$($D$)}
		\STATE \textit{HUCSPs} $\leftarrow $ \textit{HUCSPs} $\cup$ $S$;
		\ENDIF
		\STATE call \textit{\textbf{Recursive-Search($S$, $SIL$, $S$.\textit{IChain})}};
		\ENDFOR		
		\RETURN \textit{HUCSPs}
	\end{algorithmic}	
\end{algorithm}

\begin{algorithm}[ht]
	\small
	\caption{Recursive-Search}
	\label{alg:RS}
	\begin{algorithmic}[1]
		\REQUIRE 
		$S$: a sequence as the prefix; $SIL$: the SIL of all $q$-sequences; $S$.\textit{IChain}: the IChain of $S$.
		\ENSURE 
		\textit{HUCSPs}.
		
		\FOR {each instance list $il \in$ $S$.\textit{IChain}}
		\STATE obtain SIL of the $q$-sequence whose SID is \textit{il.SID};
		\STATE scan SIL to obtain the set of I-extension items \textit{Iitem}($S$);
		\STATE scan SIL to obtain the set of S-extension items \textit{Sitem}($S$);
		\ENDFOR
		
		\FOR {each item $i \in$ \textit{Iitem}($S$) }
		\STATE $S'$ $\leftarrow$ $<$$S \oplus i$$>$; \\
		$//$ \textit{The LUIP strategy}
		\IF{\textit{IEU}($S'$) $< \xi$ $\times$ $u$($D$)}
		\STATE continue;
		\ENDIF
		\STATE construct the IChain of $S'$; 
		\IF{$u$($S'$) $\ge \xi $ $\times$ $u$($D$)}
		\STATE \textit{HUCSPs} $\leftarrow $ \textit{HUCSPs} $\cup$ $S'$;
		\ENDIF
		\STATE call \textit{\textbf{Recursive-Search($S'$, $SIL$,  $S'$.\textit{IChain})}};
		\ENDFOR
		
		\FOR {each item $i \in$ \textit{Sitem}($S$) }
		\STATE $S'$ $\leftarrow$ $<$$S \otimes i$$>$; \\
		$//$ \textit{The LUIP strategy}
		\IF{\textit{IEU}($S'$) $< \xi $ $\times$ $u$($D$)}
		\STATE continue;
		\ENDIF
		\STATE construct the IChain of $S'$; 
		\IF{$u$($S'$) $\ge \xi $ $\times$ $u$($D$)}
		\STATE  \textit{HUCSPs} $ \leftarrow $ \textit{HUCSPs} $ \cup$ $S'$;
		\ENDIF
		\STATE call \textit{\textbf{Recursive-Search($S'$, $SIL$,  $S'$.\textit{IChain})}};
		\ENDFOR
	\end{algorithmic}	
\end{algorithm}

Algorithm \ref{alg:RS} presents the details of the Recursive-Search procedure that recursively excavates HUCSPs in a depth-first search manner. It takes a prefix sequence $S$, the SIL of all $q$-sequences, and the IChain of $S$ as the inputs. First, the procedure scans the IChain of $S$ and SIL to obtain the sets of I-extension and S-extension items of $S$ in $D$ (i.e., \textit{Iitem}($S$) and \textit{Sitem}($S$) (lines 1--5)). Next, the procedure processes each item $i$ in \textit{Iitem}($S$) (lines 6--16). The extension sequence of $S$ is generated by performing I-extension on $S$ with $i$ (line 7). Next, the \textit{IEU} value of $S'$ is calculated. Note that \textit{IEU}($S'$) has three components: (1) utility of $S$, (2) utility of $i$, and (3) utility of the remaining sequence with respect to $i$. \textit{IEU}($S'$) is easy to calculate because (1) is stored in the IChain of $S$, and (2) and (3) can be obtained directly from the SIL. If \textit{IEU}($S'$) $< \xi$ $\times$ $u$($D$), then $i$ is a local unpromising item for $S$ and is pruned (lines 8--9); otherwise, the IChain of $S'$ is constructed based on the IChain of $S$ (line 11). Then,  $S'$ is evaluated to determine whether it is a HUCSP (lines 12--14). Note that $u$($S'$) can be obtained from the IChain of $S'$. Finally, the procedure invokes itself recursively to generate and examine new candidate patterns with $S'$ as their prefix (line 15). The items in \textit{Sitem} can be processed using a similar procedure (lines 17--27).

\section{Experiments}   \label{sec:experiments}

This section presents the experimental results of the proposed FUCPM algorithm and its competitor, HUCP-Miner \cite{huang2016efficient}. We conducted substantial experiments on several real-world and synthetic datasets for the following purposes: (1) compare the efficiency of FUCPM and HUCP-Miner, (2) evaluate the effectiveness of the proposed pruning strategies, (3) verify the scalability of FUCPM on large-scale multi-items-based sequence datasets, and (4) compare the efficiency and mining results of FUCPM and the state-of-the-art algorithm ProUM \cite{gan2020proum} for general HUSPM. 

The experiments were performed on a personal computer equipped with an Intel(R) Core(TM) i7-8700 CPU @3.20 GHz processor and 16 GB of RAM, running a 64-bit Microsoft Windows 10 operating system. All algorithms were implemented in Java using IntelliJ IDEA. 

The remainder of this section is organized as follows. Subsection \ref{sec:Data} introduces the datasets used in the experiments. Subsections \ref{sec:Runtime}, \ref{sec:Memory}, and \ref{sec:Candidate} compare the performance of FUCPM and HUCP-Miner in terms of runtime, memory consumption, and candidate generation. Because HUCP-Miner can only handle single-item-based sequences, this group of experiments were conducted on six real-world datasets that are composed of single-item-based sequences. Subsection \ref{sec:Effectiveness} compares the performance of FUCPM and its variants to evaluate the effectiveness of the pruning strategies. Subsection \ref{sec:Scalability} reports the scalability test of FUCPM on six synthetic datasets. Finally, subsection \ref{sec:Comparison} discusses the superiority of FUCPM compared with the general HUSPM method.

\begin{table}[h]
	\caption{Features of the datasets}
	\label{features}
	\centering      
	\renewcommand{\arraystretch}{1.2}
	\begin{tabular}{|c|c|c|c|c|c|}
		\hline
		\textbf{Dataset} & \textbf{$\textit{\#Seq}$} & \textbf{$\textit{\#Item}$} & \textbf{$\textit{maxLen}$} & \textbf{$\textit{avgLen}$}  &  \textbf{$\textit{avgItem}$} \\ \hline 
		\textit{Bible} & 36,369 & 13,905 & 100 & 21.64 & 1.00 \\ \hline
		\textit{Leviathan} & 5,834 & 9,025 & 100 & 33.81 & 1.00 \\ \hline
		\textit{BMS} & 77,512 & 3,340 & 267 & 4.62 & 1.00 \\ \hline
		\textit{MSNBC} & 31,790 & 17 & 100 & 13.33 & 1.00 \\ \hline 
		\textit{Kosarak10k} & 10,000 & 10,094 & 608 & 8.14 & 1.00 \\ \hline
		\textit{FIFA} & 20,450 & 2,990 & 100 & 34.74 & 1.00 \\ \hline
		\textit{Syn10k} & 10,000 & 7,312 & 18 & 6.22 & 4.35 \\ \hline
		\textit{Syn80k} & 79,718 & 7,584 & 18 & 6.19 & 4.32 \\ \hline
		\textit{Syn160k} & 159,501 & 7,609 & 20 & 6.19 & 4.32 \\ \hline
		\textit{Syn240k} & 239,211 & 7,617 & 20 & 6.19 & 4.32 \\ \hline
		\textit{Syn320k} & 318,889 & 7,620 & 20 & 6.19 & 4.32 \\ \hline
		\textit{Syn400k} & 398,716 & 7,621 & 20 & 6.18 & 4.32 \\ \hline
	\end{tabular}
\end{table}

\subsection{Dataset Description} \label{sec:Data}

In the experiments, six real-world datasets and six synthetic datasets were used to evaluate the performance of the algorithms. The detailed characteristics of the datasets are listed in Table \ref{features}. Note that \#\textit{Seq} is the number of sequences; \#\textit{Item} is the number of distinct items; \textit{maxLen} and \textit{avgLen} are the maximum and average length of the sequences, respectively; and \textit{avgItem} is the average number of items contained in each itemset. It can be observed that the \textit{avgItem} of the six real-world datasets are all equal to 1 because these datasets are composed of single-item-based sequences. Conversely, the synthetic datasets are composed of multi-items-based sequences, where each itemset contains more than four items on average. The sources of the datasets are as follows. (1) \textit{Bible} and \textit{Leviathan} are conversions of the Bible and novel Leviathan. Each word in the books is transformed into an item, and each sentence is treated as a sequence. (2) \textit{BMS}, \textit{MSNBC}, \textit{Kosarak10k}, and \textit{FIFA} are clickstream data from an e-commerce website, the MSNBC website, the website of FIFA World Cup 98, and a Hungarian news portal, respectively. (3) The synthetic datasets (from \textit{Syn10k} to \textit{Syn400k}) are generated by the IBM Quest Dataset Generator \cite{IBM1994}. All these datasets were obtained from the open-source data mining library SPMF \footnote{http://www.philippe-fournier-viger.com/spmf/}.

\begin{figure*}[htbp]
	\centering
	\includegraphics[width=0.9\linewidth]{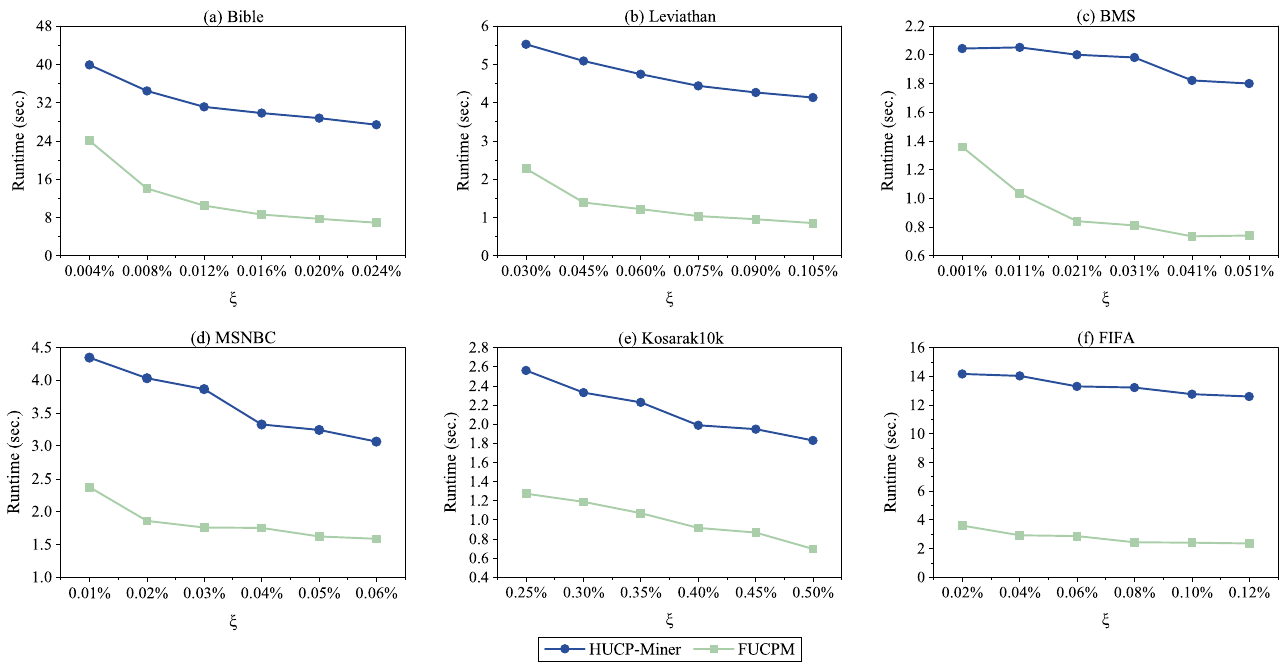}
	\caption{Runtime of HUCP-Miner and FUCPM}
	\label{runtime}
\end{figure*}

\begin{figure*}[htbp]
	\centering
	\includegraphics[width=0.9\linewidth]{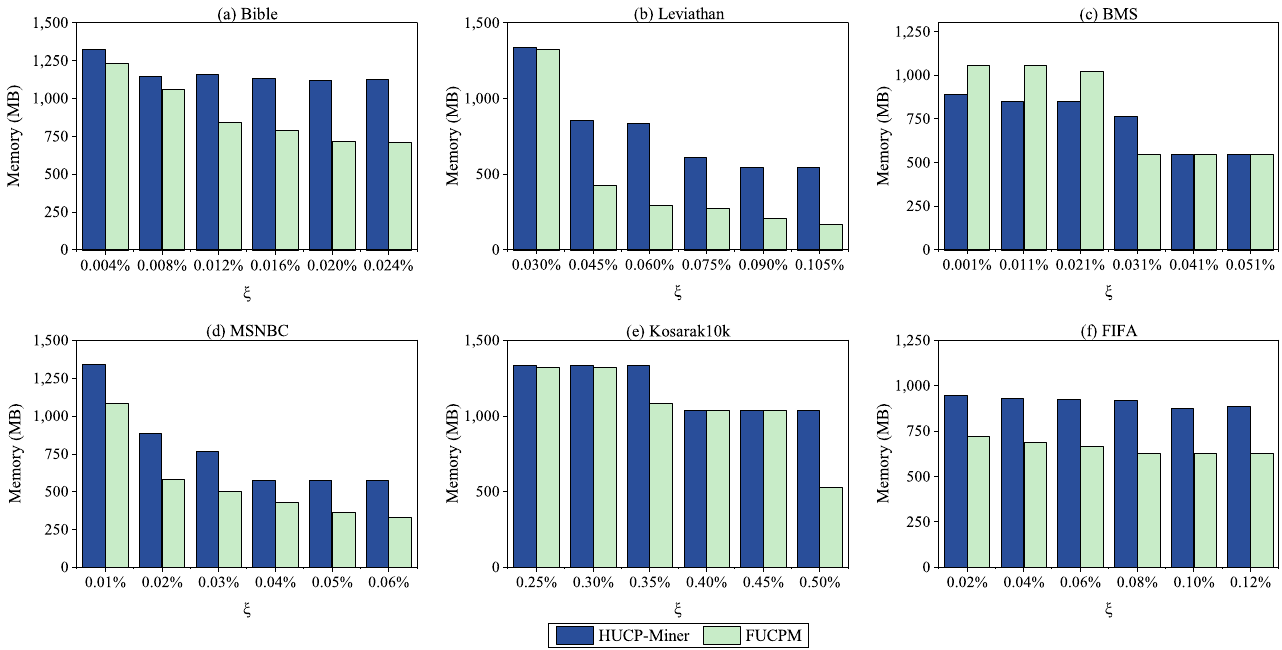}
	\caption{Memory consumption of HUCP-Miner and FUCPM}
	\label{memory}
\end{figure*}

\subsection{Runtime Analysis} \label{sec:Runtime}

Figure \ref{runtime} shows the runtime of the compared algorithms under various threshold values on the six real-world datasets. Noticeably, FUCPM has better performance than HUCP-Miner for all datasets, with an improvement of about 30\% to 75\% in terms of running speed. The advantage of FUCPM is clearer on \textit{Bible} and \textit{FIFA} datasets, wherein the average length of $q$-sequences is relatively longer. The reason for this is probably that the proposed GUIP and LUIP strategies can prune unpromising items earlier to prevent the candidate sequences from growing too long. Further, we can observe that the runtime decreases as the threshold value increases. This is because, under larger thresholds, more candidate sequences can be pruned earlier given that their upper bound value does not exceed the thresholds. Meanwhile, the runtime of FUCPM reduces more sharply than that of HUCP-Miner when the threshold value increases, especially on \textit{Bible} and \textit{BMS} datasets. In general, the prominent advantage in terms of runtime illustrates that FUCPM benefits from the proposed pruning strategies and is good at addressing the datasets that contain long sequences.

\begin{figure*}[htbp]
	\centering
	\includegraphics[width=0.9\linewidth]{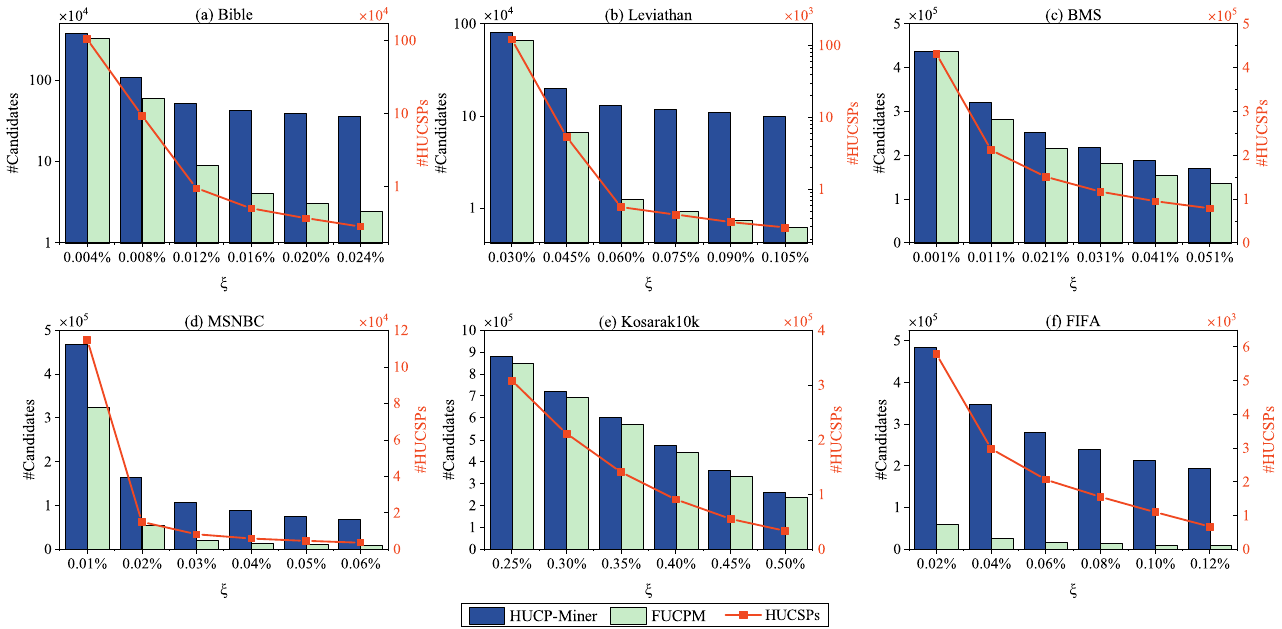}
	\caption{Candidates and discovered patterns of HUCP-Miner and FUCPM}
	\label{candidate}
\end{figure*}

\subsection{Memory Analysis} \label{sec:Memory}

The memory usage of FUCPM and HUCP-Miner is displayed in Figure \ref{memory}. It can be clearly observed that, as the threshold value increases, the memory consumption of both algorithms shows a decreasing trend. This is because the memory consumed on the instance chains of the candidates reduces when the threshold increases. We can also observe that FUCPM consumes less memory than HUCP-Miner in most cases. For the \textit{Bible} dataset, the memory consumption of HUCP-Miner is almost steady as the threshold increases, while the memory consumed by FUCPM gradually decreases. The performance gap is clearer for \textit{Leviathan}; FUCPM consumes approximately 300--500 MB less memory than HUCP-Miner under various threshold settings. However, little difference was observed in memory consumption between the two algorithms when $\xi$ = 0.030\%. A similar case also occurs for \textit{BMS} when $\xi$ = 0.041\% and 0.051\%, as well as on \textit{Kosarak10k} when $\xi$ = 0.25\%, 0.30\%, 0.40\%, and 0.45\%. This is mainly because the number of candidates generated by the two algorithms is close under such circumstances. Another interesting phenomenon is that FUCPM uses more memory than HUCP-Miner when the threshold is relatively low for \textit{BMS}. We speculate that this is because FUCPM requires more memory to store the SIL structure of a large number of $q$-sequences in \textit{BMS}, although it reduces the memory consumed on the instance chains of fewer candidates. For \textit{MSNBC} and \textit{FIFA}, FUCPM outperformed HUCP-Miner under all parameter settings. In summary, by adopting efficient pruning strategies and compact data structures, FUCPM can achieve better performance in terms of memory consumption than HUCP-Miner in most cases. 

\begin{table*}[htbp]
	\centering
	\renewcommand{\arraystretch}{1.2}
	\caption{Effective search rate of HUCP-Miner and FUCPM}
	\label{ESR}
	\begin{tabular}{|c|c|c|c|c|c|c|c|}
		\hline 
		\multirow{3}{*}{\textit{Bible}} 
		& {$\xi$} & {0.004\%} & {0.008\%} & {0.012\%} & {0.016\%} & {0.020\%} & {0.024\%} \\
		\cline{2-8}
		& {HUCP-Miner (\%)} & {27.86} & {8.67} & {1.81} & {1.17} & {0.94} & {0.79} \\
		\cline{2-8}
		& {FUCPM (\%)} & {\textbf{32.24}} & {\textbf{15.81}} & {\textbf{10.72}} & {\textbf{12.28}} &  {\textbf{12.02}} & {\textbf{11.67}} \\
		\cline{2-8}
		\hline
		
		\multirow{3}{*}{\textit{Leviathan}}
		& {$\xi$} & {0.030\%} & {0.045\%} & {0.060\%} & {0.075\%} & {0.090\%} & {0.105\%} \\
		\cline{2-8}
		& {HUCP-Miner (\%)} & {14.98} & {2.68} & {4.27} & {3.74} & {3.23} & {2.92} \\
		\cline{2-8}
		& {FUCPM (\%)} & {\textbf{18.17}} & {\textbf{8.04}} & {\textbf{4.51}} & {\textbf{4.75}} &  {\textbf{4.74}} & {\textbf{4.72}} \\
		\cline{2-8}
		\hline
		
		\multirow{3}{*}{\textit{BMS}} 
		& {$\xi$} & {0.001\%} & {0.011\%} & {0.021\%} & {0.031\%} & {0.041\%} & {0.051\%} \\
		\cline{2-8}
		& {HUCP-Miner (\%)} & {98.42} & {65.86} & {59.40} & {53.90} & {50.29} & {46.74} \\
		\cline{2-8}
		& {FUCPM (\%)} & {\textbf{98.93}} & {\textbf{74.94}} & {\textbf{69.92}}& {\textbf{64.65}} &  {\textbf{61.46}} & {\textbf{58.04}} \\
		\cline{2-8}
		\hline
		
		\multirow{3}{*}{\textit{MSNBC}} 
		& {$\xi$} & {0.01\%} & {0.02\%} & {0.03\%} & {0.04\%} & {0.05\%} & {0.06\%} \\
		\cline{2-8}
		& {HUCP-Miner (\%)} & {24.57} & {9.17} & {7.56} & {6.63} & {5.91} & {5.50} \\
		\cline{2-8}
		& {FUCPM (\%)} & {\textbf{35.30}} & {\textbf{27.00}} & {\textbf{41.94}} & {\textbf{42.30}} &  {\textbf{41.19}} & {\textbf{41.36}} \\
		\cline{2-8}
		\hline
		
		\multirow{3}{*}{\textit{Kosarak10k}} 
		& {$\xi$} & {0.25\%} & {0.30\%} & {0.35\%} & {0.40\%} & {0.45\%} & {0.50\%} \\
		\cline{2-8}
		& {HUCP-Miner (\%)} & {34.84} & {29.21} & {23.42} & {19.09} & {15.42} & {12.53} \\
		\cline{2-8}
		& {FUCPM (\%)} & {\textbf{36.18}} & {\textbf{30.51}} & {\textbf{24.65}} & {\textbf{20.32}} &  {\textbf{16.68}} & {\textbf{13.91}} \\
		\cline{2-8}
		\hline
		
		\multirow{3}{*}{\textit{FIFA}} 
		& {$\xi$} & {0.02\%} & {0.04\%} & {0.06\%} & {0.08\%} & {0.10\%} & {0.12\%} \\
		\cline{2-8}
		& {HUCP-Miner (\%)} & {1.20} & {0.86} & {0.74} & {0.65} & {0.52} & {0.35} \\
		\cline{2-8}
		& {FUCPM (\%)} & {\textbf{9.79}} & {\textbf{10.97}} & {\textbf{11.76}} & {\textbf{11.96}} &  {\textbf{10.58}} & {\textbf{7.80}} \\
		\cline{2-8}
		\hline
	\end{tabular}
\end{table*}

\subsection{Candidate Analysis} \label{sec:Candidate} 

The number of candidate patterns generated by the algorithm is an important measure of the size of the search space, which reflects the ability of the pruning strategies. In Figure \ref{candidate}, histograms are used to represent the number of generated candidate patterns and broken lines are used to represent the number of discovered HUCSPs. As shown in Figure \ref{candidate}, FUCPM generates fewer candidates than HUCP-Miner on all datasets. In particular, for \textit{MSNBC} and \textit{FIFA}, the number of candidates generated by FUCPM is significantly less than that of HUCP-Miner, which can explain why FUCPM performs better in terms of runtime and memory consumption than HUCP-Miner on these two datasets. In addition, FUCPM is evidently dominant on \textit{Bible} and \textit{Leviathan} when the threshold is relatively high. However, for \textit{BMS} and \textit{Kosarak10k}, the number of candidates generated by FUCPM did not decrease significantly compared to HUCP-Miner. We infer that this is because the data volume of these two datasets, which is calculated as \#\textit{Seq} $\times$ \textit{avgLen}, is relatively small, so the advantage of FUCPM is not quite evident.

We further introduce a new metric called the effective search rate (\textit{ESR}), which is defined as (\#\textit{HUCSPs} $\div$ \#\textit{Candidates}) $\times$ 100\%, to evaluate the search efficiency of the algorithms. Clearly, the algorithm with a higher \textit{ESR} value can reduce the unnecessary search for low-utility candidates. The \textit{ESR} values of the compared algorithms are listed in Table \ref{ESR}. We can observe that the \textit{ESR} of FUCPM is approximately 10\% higher than that of HUCP-Miner under most threshold settings for \textit{Bible}, \textit{BMS}, and \textit{FIFA}. The advantage of FUCPM is more apparent on \textit{MSNBC}, whereas for \textit{Leviathan} and \textit{Kosarak10k}, the \textit{ESR} of FUCPM is slightly higher than that of HUCP-Miner. From these discussions, we can conclude that the proposed GUIP and LUIP strategies can prune the search space more effectively than the compared method.

\begin{figure*}[h]
	\centering
	\includegraphics[width=0.9\linewidth]{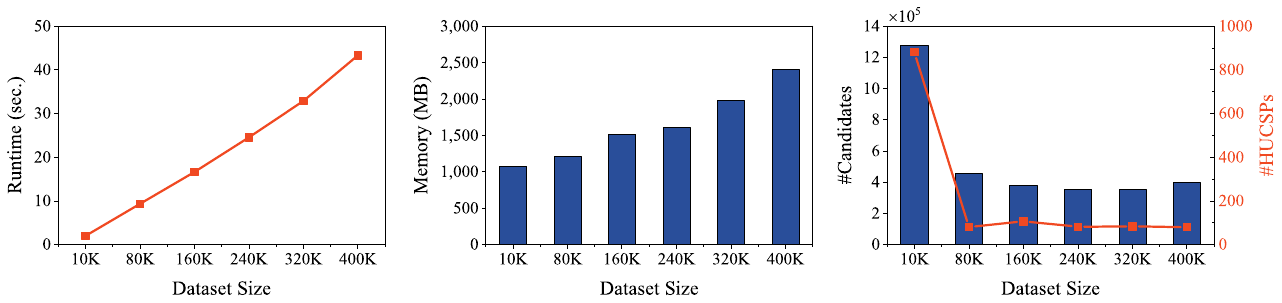}
	\caption{Scalability test of FUCPM}
	\label{scalability}
\end{figure*}

\subsection{Effectiveness of Pruning Strategies}
\label{sec:Effectiveness}

An ablation experiment was conducted on six real-world datasets to evaluate the effectiveness of the proposed pruning strategies. Figure \ref{effectiveness} shows the performance in terms of runtime and memory consumption of the FUCPM and its two variants. Note that FUCPM$_{\rm GUIP}$ and FUCPM$_{\rm LUIP}$ are the algorithms that remove GUIP and LUIP strategies from FUCPM, respectively. The minimum utility thresholds were set to 0.01\% for the six datasets. As shown in Figure \ref{effectiveness}, FUCPM is far more superior to FUCPM$_{\rm LUIP}$ in terms of runtime for all datasets. In particular, compared to FUCP$_{\rm LUIP}$, the advantage of FUCPM is much more apparent on \textit{Bible}, \textit{Leviathan}, and \textit{FIFA}, which occupy a larger \textit{avgLen}. This result demonstrates that the LUIP strategy is suitable for handling long sequences. From Figure \ref{effectiveness}, we can also observe that FUCPM uses the least memory for each dataset. However, the runtime and memory consumption of FUCPM and FUCPM$_{\rm GUIP}$ are almost equal for \textit{Leviathan}, \textit{BMS}, and \textit{MSNBC}. This is mainly because only few global unpromising items exist and can be pruned on these datasets. In conclusion, the results of the ablation experiment present the contributions of the proposed GUIP and LUIP strategies to improve the efficiency of FUCPM, and the LUIP strategy is more effective than the GUIP strategy for pruning the search space.

\begin{figure}[htbp]
	\centering
	\includegraphics[width=0.98\linewidth]{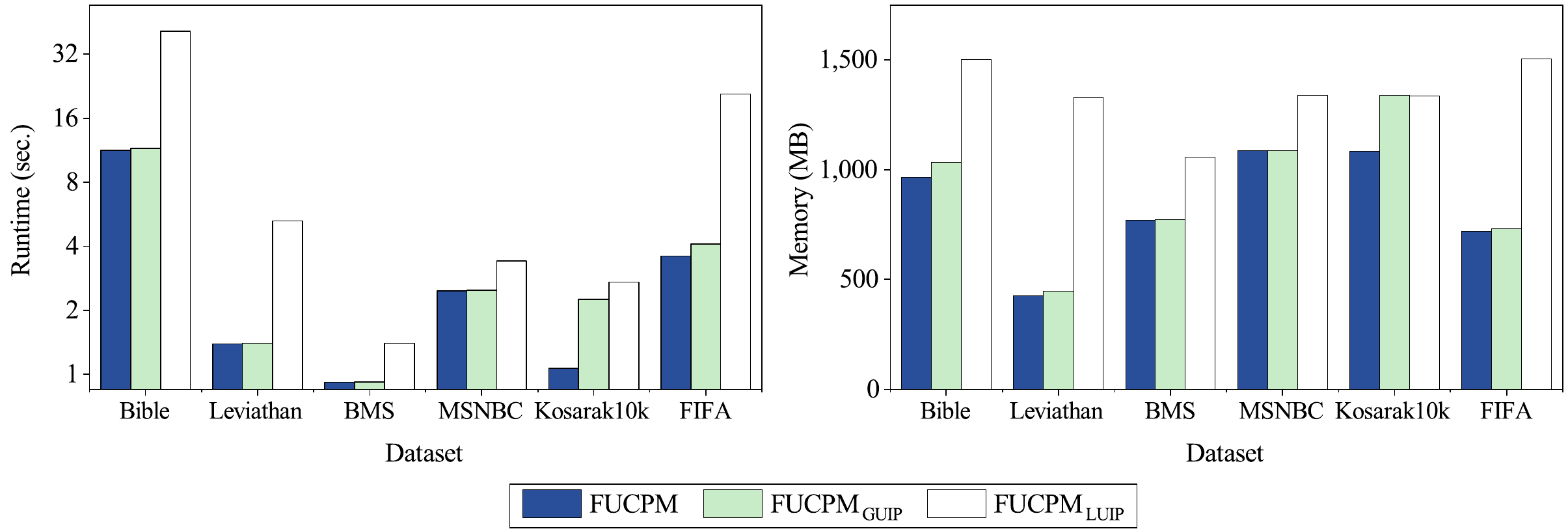}
	\caption{Effectiveness of pruning strategies}
	\label{effectiveness}
\end{figure}

\subsection{Scalability Test} \label{sec:Scalability}

The scalability of the proposed FUCPM algorithm was tested on six synthetic datasets containing multi-items-based sequences with different sizes varying from 10 K to 400 K. Figure \ref{scalability} displays the results in terms of runtime, memory consumption, and the number of candidates and HUCSPs when $\xi$ = 0.1\% for each test. Evidently, the runtime and memory consumption increase almost linearly as the dataset size increases. In addition, the number of generating candidates and HUCSPs does not increase when the dataset size becomes larger. This is because the threshold value (calculated as $\xi \times u(D)$) increases as the dataset size increases, although the utility values of some candidates may increase in a larger dataset. We can also observe that the number of candidates and HUCSPs on \textit{Syn10k} is greater. This phenomenon can be attributed to the distinctive utility distribution of this dataset. In summary, because the runtime and memory consumption are approximately linear with the dataset size, we can conclude that FUCPM is scalable to large-scale multi-items-based sequence datasets.

\subsection{Comparison with General HUSPM} \label{sec:Comparison}

To verify the efficiency of the algorithm and conciseness of the mining results, we conducted an experiment to compare FUCPM with the general HUSPM algorithm ProUM that discovers the complete set of HUSPs, including the non-contiguous ones, on six real-world datasets. The minimum utility thresholds for the six datasets were all set to 0.5\%. Table \ref{performance} presents the experimental results. We use "-" to mark the situation in which the algorithm cannot finish the mining process within 2 hours. As shown in Table \ref{performance}, the runtime and memory consumption of FUCPM are substantially less than those of ProUM. For \textit{BMS}, \textit{Kosarak10k}, and \textit{FIFA}, ProUM cannot obtain the results within 2 hours, whereas FUCPM finishes the mining process in less than 2 seconds. In addition, under the same threshold, the number of patterns discovered by FUCPM is far less than that of ProUM, which is beneficial for users who only want to obtain and analyze contiguous patterns.

\begin{table}[!ht]
	\caption{Performance of ProUM and FUCPM}
	\label{performance}
	\centering      
	\renewcommand{\arraystretch}{1.2}
	\begin{tabular}{|c|c|c|c|c|}
		\hline
		\textbf{Dataset} & \textbf{Algorithm} & \textbf{Runtime/s} & \textbf{Memory/MB} & \textbf{\#Patterns} \\ \hline 
		\multirow{2}{*}{\textit{Bible}} 
		& {ProUM} & {29.20} & {1352.71} & {2,760}  \\
		\cline{2-5}
		& {FUCPM} & {\textbf{2.12}} & {\textbf{417.35}} & {\textbf{36}} \\
		\cline{2-5}
		\hline
		\multirow{2}{*}{\textit{Leviathan}} 
		& {ProUM} & {21.87} & {1334.06} & {15,441}  \\
		\cline{2-5}
		& {FUCPM} & {\textbf{0.79}} & {\textbf{146.56}} & {\textbf{32}} \\
		\cline{2-5}
		\hline
		\multirow{2}{*}{\textit{BMS}} 
		& {ProUM} & {-} & {-} & {-}  \\
		\cline{2-5}
		& {FUCPM} & {\textbf{0.51}} & {\textbf{167.75}} & {\textbf{6,554}} \\
		\cline{2-5}
		\hline
		\multirow{2}{*}{\textit{MSNBC}} 
		& {ProUM} & {859.43} & {1345.54} & {395,626}  \\
		\cline{2-5}
		& {FUCPM} & {\textbf{1.36}} & {\textbf{323.18}} & {\textbf{290}} \\
		\cline{2-5}
		\hline
		\multirow{2}{*}{\textit{Kosarak10k}} 
		& {ProUM} & {-} & {-} & {-}  \\
		\cline{2-5}
		& {FUCPM} & {\textbf{0.69}} & {\textbf{531.71}} & {\textbf{32,976}} \\
		\cline{2-5}
		\hline
		\multirow{2}{*}{\textit{FIFA}} 
		& {ProUM} & {-} & {-} & {-}  \\
		\cline{2-5}
		& {FUCPM} & {\textbf{1.91}} & {\textbf{534.07}} & {\textbf{59}} \\
		\cline{2-5}
		\hline
	\end{tabular}
\end{table}

\section{Conclusion}   \label{sec:conclusion}

In this paper, we proposed an efficient and scalable algorithm called FUCPM to address the UCSPM problem. Specifically, two compact data structures(i.e., sequence information list and instance chain) were designed to facilitate the calculation of utility and upper bound values of the candidate patterns. To further improve the efficiency, we proposed the GUIP and LUIP strategies to prune the search space, which are based on the \textit{SWU} and the novel \textit{IEU} upper bounds, respectively. Extensive experimental results on both real-world and synthetic datasets show that FUCPM outperforms the state-of-the-art algorithm for CSPM and is scalable to large-scale and complex multi-items-based sequence datasets. Compared to conventional HUSPM, the efficiency of FUCPM is significantly improved, and the mining results are far more concise.

In the future, several interesting issues can be further researched, such as using FUCPM to discover on-shelf patterns \cite{zhang2021shelf}, and designing the distributed and parallel version of FUCPM to better handle big data \cite{zihayat2016distributed}. It is also interesting to explore the applications of UCSPM in various fields. For example, given that HUCSPs maintain the contiguous order of items in sequence data, it can be applied for next-items recommendation \cite{yap2012effective}. For some practical issues, such as text representation \cite{alias2018text} and biological sequence discovery \cite{stamoulakatou2018dla}, where the adjacent relationship of items is significant, UCSPM can also come in handy.

\ifCLASSOPTIONcompsoc
  \section*{Acknowledgments}
\else
  \section*{Acknowledgment}
\fi

This work was partially supported by National Natural Science Foundation of China (Grant No. 62002136), Guangzhou Basic and Applied Basic Research Foundation (Grant No. 202102020277), Natural Science Foundation of Guangdong Province, China (Grant No. 2020A1515010970), and Shenzhen Research Council (Grant Nos. JCYJ20200109113427092 and GJHZ20180928155209705).

\ifCLASSOPTIONcaptionsoff
  \newpage
\fi

\bibliographystyle{IEEEtran}
\bibliography{FUCPM.bib}

%\end{thebibliography}


\end{document}